\documentclass[letterpaper]{article} 
\usepackage{url}
\usepackage{amssymb}
\usepackage{amsmath}
\usepackage{natbib}
\usepackage{hyperref}
\usepackage{algorithm}
\usepackage{algorithmic}
\usepackage{graphicx}
\usepackage{epstopdf}
\usepackage{epsfig} 
\usepackage{subfigure}
\def\D{{\bf D}}
\def\d{{\bf d}}

\def\I{{\bf I}}

\def\X{{\bf X}}

\def\y{{\bf y}}
\def\z{{\bf z}}

\def\u{{\bf u}}

\def\v{{\bf v}}

\def\w{{\bf w}}
\def\0{{\bf 0}}
\def\1{{\bf 1}}

\def\bet{\mbox{\boldmath$\beta$\unboldmath}}

\def\argmin{\mathop{\rm argmin}}

\def\sgn{\mathrm{sgn}}

\newtheorem{theorem}{Theorem}
\newtheorem{lemma}{Lemma}
\newtheorem{definition}{Definition}

\newtheorem{corollary}{Corollary}

\newtheorem{remark}{Remark}

\title{A Nonconvex Approach for Structured Sparse Learning}
\date{}

\author{ Shubao Zhang ~~ Hui Qian \\
College of Computer Science and Technology  \\
Zhejiang University, Hangzhou 310027, China\\
\{bravemind, qianhui\}@zju.edu.cn \\
Zhihua Zhang \\
College of Computer Science and Engineering \\
Shanghai Jiao Tong University, Shanghai 200240, China\\
zhihua@sjtu.edu.cn
}

\begin{document}

\maketitle

\begin{abstract}


Sparse learning is an important topic in many areas such as machine learning, statistical estimation, signal processing, etc. Recently, there emerges a growing interest on structured sparse learning. In this paper we focus on the $\ell_q$-analysis optimization problem for structured sparse learning ($0< q \leq 1$). Compared to  previous work, we establish weaker conditions for exact recovery in noiseless case and a tighter non-asymptotic upper bound of estimate error in noisy case. We further prove that the nonconvex $\ell_q$-analysis optimization can do recovery with a lower sample complexity and in a wider range of cosparsity than its convex counterpart. In addition, we develop an iteratively reweighted method to solve the optimization problem under the variational framework. Theoretical analysis shows that our method is capable of pursuing a local minima close to the global minima. Also, empirical results of preliminary computational experiments illustrate that our nonconvex method outperforms both its convex counterpart and other state-of-the-art methods.
\end{abstract}

\section{Introduction}

The sparse learning problem is widely studied in many areas including machine learning, statistical estimate, compressed sensing, image processing and signal processing, etc. Typically, this problem can be defined as the following linear model
\begin{equation}\label{inverse}
  \y = {\X}{\bet} + \w,
\end{equation}
where ${\bet}\in\mathbb{R}^d$ is the  vector of regression coefficients, ${\X}\in \mathbb{R}^{m\times d}$ is a design matrix with possibly far fewer rows than columns,  $\w\in\mathbb{R}^m$ is a noise vector, and $\y\in\mathbb{R}^m$ is the noisy observation. As is well known, learning with the $\ell_1$ norm (convex relaxation of the $\ell_0$ norm), such as lasso \citep{Tibshirani1996} or basis pursuit \citep{Chenabc1998}, encourages sparse estimate of $\bet$. Recently, this approach has been extended to define structured sparsity. \citet{TibshiraniTaylor2011} proposed the generalized lasso
\begin{equation}\label{glasso}
\min_{{\bet}}  \frac{1}{2}||\y - \X\bet||_2^2 + \lambda||\D \bet||_1,
\end{equation}
which assumes that the parameter $\bet$ is sparse under a linear transformation $\D\in\mathbb{R}^{n\times d}$. An equivalent constrained version is the $\ell_1$-analysis minimization proposed by \citet{Candesabc2010}, i.e.,
\begin{equation}\label{analysis_l1}
\min_{{\bet}} ||\D \bet||_1 ~~~\textrm{ s.t.}~~~ ||\y - \X\bet||_2 \leq \epsilon ,
\end{equation}
where $\D$ is called the analysis operator. In contrast to the lasso and basis pursuit in $\D=\I$,  the generalized lasso and $\ell_1$-analysis minimization make a structured sparsity assumption so that it can explore structures on the parameter. They include several well-known models as special cases, e.g., fused lasso \citep{Tibshiraniabc2005}, generalized fused lasso \citep{Viallon2014}, edge Lasso \citep{Sharpnackabc2012}, total variation (TV) minimization \citep{Rudinabc1992}, trend filtering \citep{Kimabc2009}, the LLT model \citep{Lysakerabc2003}, the inf-convolution model \citep{ChambolleLions1997}, etc. Additionally, the generalized lasso and $\ell_1$-analysis minimization have been demonstrated to be effective and even superior over the standard sparse learning in many application problems. 

The seminal work of \citet{FanLi2001} showed that the nonconvex sparse learning holds better properties than the convex one. Motivated by that, this paper investigates the following $\ell_q$-analysis minimization ($0<q\leq 1$) problem
\begin{equation}\label{analysis_lq}
\min_{{\bet}} ||\D \bet||_q^q ~~~\textrm{ s.t.}~~~ ||\y-{\X \bet}||_2 \leq \epsilon.
\end{equation}
We consider both theoretical and computational aspects. We summary the major contributions  as follows:
\begin{enumerate}
 \item [$\bullet$] We establish weaker conditions for exact recovery in noiseless case and a tighter non-asymptotic upper bound of estimate error in noisy case. Particularly, we provide a necessary and sufficient condition guaranteeing exact recovery via the $\ell_q$-analysis minimization. To  the best of our knowledge,  our work is  the first study in this issue. 
 \item [$\bullet$] We show the advantage of the nonconvex $\ell_q$-analysis minimization ($q<1$) over its convex counterpart. Specifically, the nonconvex $\ell_q$-analysis minimization can do recovery with a lower sample complexity (on the order of $q k\log (n/k)$) and in a wider range of cosparsity.
 \item [$\bullet$]  We resort to an iteratively reweighted method to solve the $\ell_q$-analysis minimization problem. Furthermore, we prove that our method is capable to obtain a local minima close to  the global minima.
\end{enumerate}
The numerical results are consistent with the theoretical analysis. For example, the nonconvex $\ell_q$-analysis minimization indeed can do recovery with smaller sample size and in a wider range of cosparsity than the convex method. The numerical results also show that our iteratively reweighted method outperforms the other state-of-the-art methods such as NESTA \citep{Beckerabc2011}, split Bregman method \citep{Caiabc2009}, and iteratively reweighted $\ell_1$ method \citep{Candesabc2007} for the $\ell_1$-analysis minimization problem and the greedy analysis pursuit (GAP) method \citep{Namabc2011} for the $\ell_0$-analysis minimization problem ($q\rightarrow 0$ in (\ref{analysis_lq})) .

\subsection{Related Work}

\citet{Candesabc2010} studied the $\ell_1$-analysis minimization problem in the setting that the observation is contaminated with stochastic noise and the analysis vector $\D\bet$ is approximately sparse. They provided a $\ell_2$ norm estimate error bounded by $C_0\epsilon + C_1 k^{-1/2}||\D \bet-(\D \bet)(k)||_1$ under the assumption that $\X$ obeys the D-RIP condition $\delta_{2k}<0.08$ or $\delta_{7k}<0.6$ and $\D$ is a Parseval tight frame \footnote{A set of vectors $\{\d_k\}$ is a frame of $\mathbb{R}^d$ if there exist constants $0<A\leq B <\infty$ such that $ \forall \bet\in\mathbb{R}^d, ~~ A||\bet||_2^2 \leq ||D\bet||_2^2 \leq B||\bet||_2^2, $ where $\{\d_k\}$ are the columns of $\D^{T}$. When $A=B=1$, the columns of $\D^{T}$ form a Parseval tight frame and $\D^{T}\D=\boldsymbol{I}$.}. \citet{Namabc2011}  studied the $\ell_1$-analysis minimization problem in the setting that there is no noise and the analysis vector $\D\bet$ is sparse. They showed that a null space property with sign pattern is necessary and sufficient to guarantee exact recovery. \citet{Liuabc2012} improved the analysis in \citep{Candesabc2010}. They established an estimate error bound similar to the one in \citep{Candesabc2010} for the general frame case. And for the Parseval frame case, they provided a weaker D-RIP condition $\delta_{2k}<0.2$.

\citet{TibshiraniTaylor2011} proposed the generalized lasso and developed a LARS-like algorithm pursuing its solution path. \citet{Vaiterabc2013} conducted a robustness analysis of the generalized lasso against noise. \citet{Liuabc2013} derived an estimate error bound for the generalized lasso under the assumption that the condition number of $\D$ is bounded. Specifically, a $\ell_2$ norm estimate error bounded by $C\lambda + ||({\X}^T{\X})^{-1}{\X}^T\w||_2$ is provided. \citet{NeedellWard2013}  investigated the total variation minimization. They proved that for an image $\bet\in\mathbb{R}^{N\times N}$, the TV minimization can stably recover it with estimate error less than $C \log(\frac{N^2}{k}) (\epsilon + ||\D\bet-(\D\bet)(k)||_1 /\sqrt{k} )$ when the sampling matrix satisfies the RIP of order $k$.

So far, all the related works discussed above consider convex optimization problem. \citet{Aldroubiabc2012} first studied the nonconvex $\ell_q$-analysis minimization problem (\ref{analysis_lq}). They established estimate error bound using the null space property and restricted isometry property respectively. For the Parseval frame case, they showed that the D-RIP condition $\delta_{7k}<\frac{6-3(2/3)^{2/q-2}}{6-(2/3)^{2/q-2}}$ is sufficient to guarantee stable recovery. \citet{LiLin2014} showed that the D-RIP condition $\delta_{2k}<0.5$ is sufficient to guarantee the success of $\ell_q$-analysis minimization. In this paper, we significantly improve the analysis of $\ell_q$-analysis minimization. For example, we provide a weaker D-RIP condition $\delta_{2k}<\frac{\sqrt{2}}{2}$. Additionally, we show the advantage of the nonconvex $\ell_q$-analysis minimization over its convex counterpart.

\section{Preliminaries}

Throughout this paper, $\mathbb{N}$ denotes the natural number. $\lfloor \cdot\rfloor$ denotes the rounding down operator. The $i$-th entry of a vector ${\bet}$ is denoted by $\beta_i$. The \emph{best $k$-term approximation} of a vector ${\bet}\in\mathbb{R}^d$ is obtained by setting its $d-k$ insignificant components to zero and denoted by ${\bet}(k)$. The $\ell_q$ norm of a vector ${\bet}\in\mathbb{R}^d$ is defined as $||{\bet}||_q = (\sum_{i=1}^d |\beta_i|^q)^{1/q}$ \footnote{$||\bet||_q$ for $0<q<1$ is not a norm, but $d(\u,\v)=||\u-\v||_q^q$ for $\u,\v\in \mathbb{R}^d$ is a metric.} for $0<q<\infty$. When $q$ tends to zero, $||{\bet}||_q^q$ is the $\ell_0$ norm $||\bet||_0$ used to measure the \emph{sparsity} of $\bet$.  $\sigma_{k}(\bet)_q = \inf_{\z\in\{\z\in\mathbb{R}^d: ||\z||_0\leq k\}} ||\bet-\z||_q$ denotes the best $k$-term approximation error of $\bet$ with the $\ell_q$ norm. The $i$-th row of a matrix ${\D}$ is denoted by $\D_{i.}$. $\sigma_{max}(\D)$ and $\sigma_{min}(\D)$ denote  the maximal and minimal nonzero singular value of $\D$, respectively. Let $\kappa=\frac{\sigma_{max}(\D)}{\sigma_{min}(\D)}$, and  $\textrm{Null} (\X)$ denote the null space of $\X$.

Now we introduce some concepts related to the $\ell_q$-analysis minimization problem (\ref{analysis_lq}). The number of zeros in the analysis vector $\D \bet$ is refered to as \emph{cosparsity} \citep{Namabc2011}, and defined as $l:= n-||\D \bet||_0$. Such a vector ${\bet}$ is said to be $l$-cosparse. The $support$ of a vector ${\bet}$ is the collection of indices of nonzeros in the vector, denoted by $T:=\{i : \beta_i\neq0\}$. $T^c$ denotes the complement of $T$. The indices of zeros in the analysis vector $\D \bet$ is defined as the $cosupport$ of ${\bet}$, and denoted  by $\Lambda:=\{j : \langle \D_{j.},{\bet} \rangle=0\}$. The submatrix $\D_{T}$ is constructed by replacing the rows of $\D$ corresponding to $T^c$ by zero rows. Denote $\D_T\bet=(\D\bet)_T$. Based on these concepts, we can see that a $l$-cosparse vector $\bet$ lies in the subspace $\mathcal{W}_{\Lambda} :=\{{\bet}:\D_{\Lambda}{\bet}=\boldsymbol{0}, |\Lambda|= l\} = \textrm{Null}(\D_{\Lambda})$. Here $|\Lambda|$ is the cardinality of $\Lambda$.

In our analysis below, we use the notion of $\mathcal{A}$-RIP \citep{BlumensathDavies2008}.
\begin{definition}
[{$\mathcal{A}$-restricted isometry property}] A matrix $\boldsymbol{\Phi}\in\mathbb{R}^{m\times d}$ obeys the $\mathcal{A}$-restricted isometry property with constant $\delta_{\mathcal{A}}$ over any subset $\mathcal{A}\in\mathbb{R}^d$, if $\delta_{\mathcal{A}}$ is the smallest quantity satisfying
\[ (1-\delta_{\mathcal{A}})||{\v}||_2^2 \leq ||\boldsymbol{\Phi}{\v}||_2^2 \leq (1+\delta_{\mathcal{A}})||{\v}||_2^2 \]
for all ${\v}\in\mathcal{A}$.
\end{definition}
Note that RIP \citep{CandesTao2004}, D-RIP \citep{Candesabc2010} and $\boldsymbol{\Omega}$-RIP \citep{Giryesabc2013} are special instances of the $\mathcal{A}$-RIP with different choices of the set $\mathcal{A}$. For example, when choosing $\mathcal{A}=\{{\D\v}: \v\in\mathbb{R}^d, ||\v||_0\leq k\}$ and $\mathcal{A}=\{{\v}: \v\in\mathbb{R}^d, \D_{\Lambda}\v = \boldsymbol{0}, |\Lambda|\geq l\}$, the corresponding $\mathcal{A}$-restricted isometries are  D-RIP and $\Omega$-RIP, respectively. It has been verified that any random matrix $\boldsymbol{\Phi}$ holds the $\mathcal{A}$-restricted isometry property with overwhelming probability provided that the number of samples depends logarithmically on the number of subspaces in $\mathcal{A}$ \citep{BlumensathDavies2008}.

\section{Main Results}

In this section, we present our main theoretical results pertaining to the ability of $\ell_q$-analysis minimization to estimate (approximately) cosparse vectors with and without noise.

\subsection{Exact Recovery in Noiseless Case}

A well-known necessary and sufficient condition guaranteeing the success of basis pursuit is the null space property \citep{Cohenabc2009}. Naturally, we define a null space property adapted to $\D$ (D-NSPq) of order $k$ \citep{Aldroubiabc2012} for the $\ell_q$-analysis minimization. That is,
\begin{equation} \label{D-NSPq}
  \forall \boldsymbol{v}\in \textrm{Null} (\X)/\{\boldsymbol{0}\}, \forall |T|\leq k,  ||\boldsymbol{D}_T\boldsymbol{v}||_q^q < ||\boldsymbol{D}_{T^c}\boldsymbol{v}||_q^q .
\end{equation}

\begin{theorem} \label{exact_nsp}
Let $\bet\in\mathbb{R}^d$ with cosupport $\Lambda$, $||\D\bet||_0=k$, and $\y=\X\bet$. Then $\bet$ is the unique minimizer of the $\ell_q$-analysis minimization (\ref{analysis_lq}) with $\epsilon=0$ if and only if $\X$ satisfies the D-NSPq (\ref{D-NSPq}) relative to $\Lambda^c$.
\end{theorem}
Letting the set $\Lambda$ ($|\Lambda|=l$) vary, the following result is a corollary of Theorem \ref{exact_nsp}.
\begin{corollary}
Given a matrix $\X\in\mathbb{R}^{m\times d}$ and $\y=\X\bet$, the $\ell_q$-analysis minimization (\ref{analysis_lq}) with $\epsilon=0$ recovers every $l$-cosparse vector $\bet\in\mathbb{R}^d$ as a unique minimizer if and only if $\X$ satisfies the D-NSPq (\ref{D-NSPq}) of order $n-l$.
\end{corollary}

This corollary establishes a necessary and sufficient condition for exact recovery of all $l$-cosparse vectors via the $\ell_q$-analysis minimization. It also implies that for every $\y=\X\bet$ with $l$-cosparse $\bet$, the $\ell_q$-analysis minimization actually solves the $\ell_0$-analysis minimization when the D-NSPq of order $n-l$ holds. Based on the D-NSPq, the following corollary shows that the nonconvex $\ell_q$-analysis minimization is not worse than its convex counterpart.

\begin{corollary}
 For $0<q_1 <q_2\leq 1$, the sufficient condition for exact recovery via the $\ell_{q_2}$-analysis minimization is also sufficient for exact recovery via the $\ell_{q_1}$-analysis minimization.
\end{corollary}

It is hard to check the D-NSPq (\ref{D-NSPq}). The following theorem provides a sufficient condition for exact recovery using the $\mathcal{A}$-RIP.
\begin{theorem} \label{exact_rip}
Let $\bet\in\mathbb{R}^d$, $||\D\bet||_0=k$, and $\y=\X\bet$.  Assume that $\D\in\mathbb{R}^{n\times d}$ has full column rank, and its condition number is upper bounded by $\kappa< \sqrt{\frac{2\rho+1+\sqrt{4\rho+1}}{2\rho}}$. If $\X\in\mathbb{R}^{m\times d}$ satisfies the $\mathcal{A}$-RIP over the set $\mathcal{A}=\{{\D\v}: ||\v||_0 \leq (t^q+1)k\}$ with $k,t^q k\in\mathbb{N}, t>0,q\in(0,1]$, i.e.,
\begin{equation}\label{noiseless_cond}
\delta_{(t^q+1)k} < \frac{\rho(1-\kappa^4) + \kappa^2 \sqrt{4\rho+1}}  { \rho (\kappa^2+1)^2 + \kappa^2 }
\end{equation}
with $\rho= t^{q-2}/4$, then $\bet$ is the unique minimizer of the $\ell_q$-analysis minimization
 (\ref{analysis_lq}) with $\epsilon=0$.
\end{theorem}

This theorem says that although the $\ell_q$-analysis minimization is a nonconvex optimization problem with many local minimums, one still can find the global optimum under the condition (\ref{noiseless_cond}). As pointed out by
\citet{BlanchardThompson2009}, the higher-order RIP condition, just as (\ref{noiseless_cond}), is easier to be satisfied by a larger subset of matrix ensemble such as Gaussian random matrices. Thus, our result is meaningful both theoretically and practically. 

It is easy to verify that the right-hand side of the condition (\ref{noiseless_cond}) is monotonically decreasing with respect to $q\in(0,1]$ when $t\geq 1$. Therefore, in terms of the $\mathcal{A}$-RIP constant $\delta_{(t^q+1)k}$ with order more than $2k$, the condition (\ref{noiseless_cond}) is relaxed if we use the $\ell_q$-analysis minimization ($q<1$) instead of the $\ell_1$-analysis minimization. A resulted benefit is that the nonconvex $\ell_q$-analysis minimization allows more sampling matrices to be used than its convex counterpart in compressed sensing.  Given a $\rho$, a larger condition number $\kappa$ will make the condition (\ref{noiseless_cond}) more restrictive, because the value of the inequality's right-hand side becomes smaller. In other words, an analysis operator with a too large condition number could let the $\ell_q$-analysis  minimization fail to do recovery. This provides hints on the evaluation of the analysis operator. For example, it is reasonable to choose a tight frame as the analysis operator in some signal processing applications. When $q$ tends to zero, the following result is straightforward.

\begin{corollary}
Let $\bet\in\mathbb{R}^d$, $\y=\X\bet$, and $||\D\bet||_0=k$. Assume that $\delta_{2k} < \frac{\rho(1-\kappa^4) + \kappa^2 \sqrt{4\rho+1}}  { \rho (\kappa^2+1)^2 + \kappa^2 }$ with $\rho=t^{-2}/4$. Then there is some small enough $q>0$ such that the minimizer of the $\ell_q$-analysis minimization problem (\ref{analysis_lq}) with $\epsilon=0$ is exactly $\bet$.
\end{corollary}

\begin{table}[t]\setlength{\tabcolsep}{1.6pt} \label{cond_table}
\caption{Different Sufficient Conditions}
\label{tab:1}
\vskip 0.15in
\begin{center}
\begin{small}
\begin{sc}
\begin{tabular}{lccccr}
\hline
$q$ & $t$ &  $\kappa$ &  Recovery condition  \\
\hline
1  &  1   &  1    &  $\delta_{2k} < \frac{\sqrt{2}}{2}$    \\
$\frac{1}{2}$  &  1   &     1    &  $\delta_{2k} < \frac{\sqrt{2}}{2}$    \\
1  &  2   &     1    &  $\delta_{3k} < \sqrt{ \frac{2}{3} }$    \\
$\frac{1}{2}$  &  4   &     1    &  $\delta_{3k} < \sqrt{ \frac{8}{9} }$    \\
1  &  6   &      1    &  $\delta_{7k} < \sqrt{ \frac{6}{7} }$    \\
$\frac{1}{2}$  &  36   &     1    &  $\delta_{7k} < \sqrt{ \frac{216}{217} }$    \\
\hline
\end{tabular}
\end{sc}
\end{small}
\end{center}
\vskip -0.1in
\end{table}

\begin{remark}
 \rm In the case $\D=\I$ and $q=1$, the condition (\ref{noiseless_cond}) is the same as the one of Theorem 1.1 in \citep{CaiZhang2014} which is a sharp condition for the basis pursuit problem. Table \ref{cond_table} shows several sufficient conditions for exact recovery via the $\ell_q$-analysis optimization. Compared to  previous work, our results promote a significant improvement. For example, for the $\ell_1$-analysis minimization, our condition $\delta_{2k} <\frac{\sqrt{2}}{2}$  is weaker than the conditions $\delta_{2k} <0.08$ in \citep{Candesabc2010}, $\delta_{2k} <0.2$ in \citep{Liuabc2012}, $\delta_{2k}<0.47$ in \citep{Linabc2013}, $\delta_{2k} <0.49$ in \citep{LiLin2014}; and our $\delta_{7k} <\sqrt{\frac{6}{7}}$ is weaker than $\delta_{7k} <0.6$ in \citep{Candesabc2010} and \citep{Aldroubiabc2012}, $\delta_{7k} <0.687$ in \citep{Linabc2013}. While for the $\ell_q$-analysis minimization ($q<1$), the D-RIP conditions $\delta_{7k}<\frac{6-3(2/3)^{2/q-2}}{6-(2/3)^{2/q-2}}$ \citep{Aldroubiabc2012} and $\delta_{2k}<0.5$ \citep{LiLin2014} are both stronger than our condition (\ref{noiseless_cond}). Note that above results all consider the Parseval tight frame case ($\kappa=1$).
\end{remark}

\subsection{Stable Recovery in Noisy Case}

Now we consider the case that the observation is contaminated with stochastic noise ($\epsilon\neq 0$) and the analysis vector $\D\bet^{*}$ is approximately sparse. This is of great interest for many applications. Our goal is to provide estimate error bound between the population parameter $\bet^{*}$ and the minimizer $\hat{\bet}$ of the $\ell_q$-analysis minimization (\ref{analysis_lq}).
\begin{theorem} \label{stable_rip}
Let $\bet^{*}\in\mathbb{R}^d$, $\y=\X\bet^{*}+\w$, and $||\w||\leq \epsilon$. Assume that $\D\in\mathbb{R}^{n\times d}$ has full column rank, and its condition number is upper bounded by $\kappa< \sqrt{\frac{2\rho+1+\sqrt{4\rho+1}}{2\rho}}$. If $\X\in\mathbb{R}^{m\times d}$ satisfies the $\mathcal{A}$-RIP over the set $\mathcal{A}=\{{\D\v}: ||\v||_0 \leq (t^q+1)k \}$ with $k, t^q k\in\mathbb{N}, t>0, q\in(0,1] $, i.e.,
\begin{equation} \label{noisy_cond}
\delta_{(t^q+1)k} < \frac{\rho(1-\kappa^4) + \kappa^2 \sqrt{4\rho+1}} { \rho (\kappa^2+1)^2 + \kappa^2 }
\end{equation}
with $\rho= 4^{1/q-2}t^{q-2}$, then the minimizer $\hat{\bet}$ of the $\ell_q$-analysis minimization problem (\ref{analysis_lq}) obeys
\begin{align*}
  & || \boldsymbol{D}\hat{\bet} - \boldsymbol{D}\bet^{*} ||_q^q \leq  2 c_1^q k^{1-q/2}  \epsilon^q +  2(2c_2^q +1) \sigma_{k}(\D\bet^{*})_q^q,  \\
  & ||\hat{\bet} - \bet^{*} ||_2 \leq  \frac{2 c_1 }{\sigma_{min}(\boldsymbol{D}) } \epsilon + \frac{ 2^{1/q}  ( 2c_2 + 1) }{\sigma_{min}(\boldsymbol{D}) } \frac{ \sigma_{k}(\D\bet^{*})_q }{k^{1/q-1/2}},
\end{align*}
where
\begin{align*}
 & c_0 = (\frac{1}{2}-\mu)^2 (1+\delta_{(t^q+1)k}) \kappa^2 - \frac{1}{4}(1-\delta_{(t^q+1)k})     \\
 &~~~~~~~~ + \rho \mu^2(\kappa^2(1+\delta_{(t^q+1)k})-(1-\delta_{(t^q+1)k})),  \\
 & c_1 = \frac{ 2\kappa(\mu-\mu^2)\sqrt{1+\delta_{(t^q+1)k}} \sigma_{max}(\D) }{-c_0}, \\
 & c_2 = \frac{ 2\rho \mu^2(\kappa^2(1+\delta_{(t^q+1)k})-(1-\delta_{(t^q+1)k})) } {-c_0}  \\
 &~~~~~~~~ + \frac{ \sqrt{ - c_0 \rho \mu^2 (\kappa^2(1+\delta_{(t^q+1)k}) -(1-\delta_{(t^q+1)k}) ) } }{ -c_0},
\end{align*}
and $\mu>0$ is a constant depending on $\rho$ and $\kappa$.
\end{theorem}
The $\ell_2$ error bound shows that the $\ell_q$-analysis optimization can stably recover the approximately cosparse vector in presence of noise. Again, we can see that a too ill-conditioned analysis operator leads to bad performance. Additionally, a $\ell_q$ error bound of the difference between $\bet^{*}$ and $\hat{\bet}$ in the analysis domain is provided, which will be used to show the advantage of the $\ell_q$-analysis minimization in the next subsection.

The linear model (\ref{inverse}) with Gaussian noise is of particular interest in machine learning and signal processing. 
Lemma 1 in \citep{CaiTabc2009} shows that the noise vector $\w\sim N(0,\sigma^2 \I)$ is upper bounded  by $\sigma \sqrt{m+2\sqrt{m\log m}}$. The following result is thus evident.
\begin{corollary}
 If the matrix $\X\in\mathbb{R}^{m\times d}$ satisfies the $\mathcal{A}$-RIP condition (\ref{noisy_cond}) and the noise vector $\w\sim N(0,\sigma^2 \I)$, then the minimizer $\hat{\bet}$ of (\ref{analysis_lq}) satisfies
\begin{align*}
& ||\hat{\bet} - \bet^{*} ||_2 \leq  \frac{2 c_1 }{\sigma_{min}(\boldsymbol{D}) } \sigma \sqrt{m+2\sqrt{m\log m}}  \\
&~~~~~~~~~~~~~~~~~~~~~~~~~~~~ + \frac{ 2^{1/q}  ( 2c_2 + 1) }{\sigma_{min}(\boldsymbol{D}) } \frac{ \sigma_{k}(\D\bet)_q }{k^{1/q-1/2}}
\end{align*}
 with probability at least $1-\frac{1}{m}$.
\end{corollary}

\subsection{Benefits of Nonconvex $\ell_q$-analysis Minimization}

The advantage of the nonconvex $\ell_q$-analysis minimization over its convex counterpart is two-fold: the nonconvex approach can do recovery with a lower sample complexity and in a wider range of cosparsity. 

The following theorem is a natural extension of Theorem 2.7 of \citet{Fourcartabc2010} in which $\D=\I$.
\begin{theorem}
 Let $m,n,k\in\mathbb{N}$ with $m,k<n$. Suppose that a matrix $\X\in\mathbb{R}^{m\times d}$, a linear operator $\D\in\mathbb{R}^{n\times d}$ and a decoder $\bigtriangleup :\mathbb{R}^m\rightarrow \mathbb{R}^d$ solving $\y=\X\bet$ satisfy for all $\bet\in\mathbb{R}^d$,
 \[
  ||\D\bet- \bigtriangleup (\X\bet)||_q^q \leq C \sigma_{k}(\D\bet)_q^q
 \]
with some constant $C>0$ and some $0<q\leq 1$. Then the minimal number of samples $m$ obeys
\[
 m \geq C_1 q k \log (n/4k)
\]
with $k=||\D\bet||_0$ and $C_1=1/(2\log(2C+3))$.
\end{theorem}
Define the decoder $\bigtriangleup(\X\bet) := \D \bigtriangleup_0(\y)$ with $\bigtriangleup_0(\y) := \argmin_{\bet,\y=\X\bet} ||\D\bet||_q^q$. Combining with the $\ell_q$ error bound in Theorem \ref{stable_rip}, we attain the following result. 
\begin{corollary} \label{cor:5}
 To recover the population parameter $\bet^{*}$, the minimal number of samples $m$ for the $\ell_q$-analysis minimization must obey
 \[
 m \geq C_2 q k \log (n/4k),
\]
where $k=||\D\bet^{*}||_0$ and $C_2=1/(2\log(8c_2^q+7))$ ($c_2$ is the constant in Theorem \ref{stable_rip}).
\end{corollary}
\begin{remark}
\rm In our analysis of the estimate error above, we used the $\mathcal{A}$-RIP over the set $\mathcal{A}=\{{\D\v}: ||\v||_0 \leq (t^q+1)k \}$, i.e., the D-RIP. As pointed out by \citet{Candesabc2010}, random matrices with Gaussian, subgaussian, or Bernoulli entries satisfy the D-RIP with sample complexity on the order of $k\log (n/k)$. It is consistent with Corollary~\ref{cor:5} in the case $q=1$. However, we see that the $\ell_q$-analysis minimization can have a lower sample complexity than the $\ell_1$-analysis minimization. Additionally, to guarantee the uniqueness of a $l$-cosparse solution of $\ell_0$-analysis minimization, the minimal number of samples required should satisfy the following condition:
\begin{equation} \label{sampling_lb}
   m\geq 2\cdot \max_{|\Lambda|\geq l} \textrm{dim}(\mathcal{W}_{\Lambda}),
\end{equation}
where $\mathcal{W}_{\Lambda}=\textrm{Null}(\D_{\Lambda})$. Please refer to \citet{Namabc2011} for more details. Therefore, the sample complexity of $\ell_q$-analysis minimization is lower bounded by $2\cdot \max_{|\Lambda|\geq l} \textrm{dim}(\mathcal{W}_{\Lambda})$.
\end{remark}

The condition (\ref{noiseless_cond}) guarantees that cosparse vectors can be exactly recovered via the $\ell_q$-analysis minimization. Define $S_q$ ($0<q\leq1$) as the largest value of the sparsity $S\in\mathbb{N}$ of the analysis vector $\D\bet$ such that the condition (\ref{noiseless_cond}) holds for some $t^q\in\frac{1}{S}\mathbb{N}$. The following theorem indicates the relationship between $S_q$ with $q<1$ and $S_1$ with $q=1$.
\begin{theorem} \label{thm:5}
Suppose that there exist $S_1\in\mathbb{N}$ and $t\in\frac{1}{S_1}\mathbb{N}$ such that
 \[
  \delta_{(t+1)S_1} <\frac{\rho(1-\kappa^4) + \kappa^2 \sqrt{4\rho+1}} { \rho (\kappa^2+1)^2 + \kappa^2 }
\]
with $\rho= \frac{1}{4} t^{-1}$.
Then there exist $S_q\in\mathbb{N}$ and $l^q\in\frac{1}{S_q}\mathbb{N}$ obeying
\begin{equation}\label{SqS1}
 S_q = \Big\lfloor \frac{t+1}{t^{\frac{q}{2-q}}+1} S_1 \Big\rfloor
\end{equation}
such that $(t+1)S_1 = (l^q+1)S_q$ and
\[
 \delta_{(l^q+1)S_q} < \frac{\rho(1-\kappa^4) + \kappa^2 \sqrt{4\rho+1}} { \rho (\kappa^2+1)^2 + \kappa^2 }
\]
with $\rho= \frac{1}{4} l^{q-2}$.

\end{theorem}
It can be verified that Theorem~\ref{thm:5} also holds for the condition (\ref{noisy_cond}). The equation (\ref{SqS1}) states that the $\ell_q$-analysis minimization with $q<1$ can do recovery in a wider range of cosparsity than the $\ell_1$-analysis minimization. For example, if $\delta_{5S_1}<\frac{2\sqrt{5}}{5}$, then the $\ell_{\frac{2}{3}}$-analysis minimization can recover a vector $\bet$ with $||\D\bet||_0=S_{\frac{2}{3}}=\lfloor \frac{5}{3} S_1 \rfloor $.

\section{Iteratively Reweighted Method for $\ell_q$-analysis Minimization}

The iteratively reweighted method is a classical approach to deal with the $\ell_q$ norm related optimization problem; see \citep{GorodnitskyRao1997,ChartrandYin2008,Daubechiesabc2010,Lu2014}. Inspired by them, we develop an iteratively reweighted method to solve the $\ell_q$-analysis optimization.
We reformulate (\ref{analysis_lq}) as the following unconstrained optimization problem:
\begin{equation}\label{unlq}
 \min_{{\bet}} \Big\{  \frac{1}{2}||\y-{\X}{\bet}||_2^2  + \lambda ||\D \bet||_q^q  \Big\}.
\end{equation}
It is hard to solve (\ref{unlq}) directly due to the nonsmoothness and nonseparability of the $\ell_q$ norm term. 
We provide a way to deal with the $\ell_q$ norm under the variational framework.

Note that the function $||\D{\bet}||_q^q$ is concave with respect to $|\D{\bet}|^{\alpha}=(|\boldsymbol\D_{1.} \bet|^{\alpha}, \ldots,|\boldsymbol\D_{n.} \bet|^{\alpha})^T$ for $\alpha\geq 1$. Thus there exists a variational upper bound of $||\D{\bet}||_q^q$. Given a positive vector $\boldsymbol{\eta} = (\eta_1,\ldots, \eta_n)^T$, we have the following variational formulation,
\begin{align*}
& ||\D{\bet}||_q^q = \sum_{i=1}^n (|\boldsymbol\D_{i.} \bet|^{\alpha})^{\frac{q}{\alpha}}  \\
&~~~~~~= \min_{\boldsymbol{\eta} > \boldsymbol{0}} \Big\{ J_{\alpha} \triangleq \frac{q}{\alpha}  \sum_{i=1}^n \Big( \eta_i|\boldsymbol\D_{i.} \bet|^{\alpha} + \frac{\alpha-q}{q} \frac{1}{{\eta_i}^{\frac{q}{\alpha-q}}} \Big) \Big\}
\end{align*}
for $\alpha\geq 1$ and $0<q\leq 1$. The function $J_{\alpha}$ is jointly convex in $({\bet},\boldsymbol{\eta})$. Its minimum is achieved at $\eta_i=1/|\boldsymbol\D_{i.} \bet|^{\alpha-q}$, $i=1,\ldots,n$. However, when ${\bet}$ is orthogonal to some $\D_{i.}$, the weight vector $\boldsymbol{\eta}$ may include infinite components. To avoid an infinite weight, we add a smoothing term $q/\alpha \sum_{i=1}^n \eta_i\varepsilon^{\alpha}$ $(\varepsilon\geq 0$) to $J_{\alpha}$.

Using the above variational formulation, we obtain an approximation of the problem (\ref{unlq}) as
\begin{align}\label{coirls}
& \min_{{\bet}} \bigg\{ F({\bet},\varepsilon) \triangleq  \min_{\boldsymbol{\eta}>\boldsymbol{0}} \frac{1}{2}||\y-{\X}{\bet}||_2^2  \\
&~~~~~ +\frac{\lambda q}{\alpha} \sum_{i=1}^n  \Big[ \eta_i(|\boldsymbol\D_{i.}{\bet}|^{\alpha}+\varepsilon^{\alpha}) {+} \frac{\alpha{-}q}{q} \frac{1}{{\eta_i}^{\frac{q}{\alpha{-}q}}} \Big] \bigg\}.\nonumber
\end{align}
We then develop an alternating minimization algorithm, which consists of three steps. The first step calculates $\boldsymbol{\eta}$ with ${\bet}$ fixed via
 \[
\boldsymbol{\eta}^{(k)} = \argmin_{\boldsymbol{\eta}>\boldsymbol{0}} \Big\{ \sum_{i=1}^n  \Big[ \eta_i(|\boldsymbol\D_{i.}\bet^{(k{-}1)}|^{\alpha} {+} \varepsilon^{\alpha}) {+} \frac{\alpha {-} q}{q} \frac{1}{{\eta_i}^{\frac{q}{\alpha {-} q}}} \Big] \Big\},
\]
which has a closed form solution. The second step calculates ${\bet}$ with $\boldsymbol{\eta}$ fixed via
\begin{equation*}\label{rwl2}
 {{\bet}}^{(k)} = \argmin_{{\bet}\in \mathbb{R}^d} \Big\{\frac{1}{2}||\y-{\X}{\bet}||_2^2 + \frac{\lambda q}{\alpha} \sum_{i=1}^n \eta_i^{(k)} |\boldsymbol\D_{i.}{\bet}|^{\alpha} \Big\},
\end{equation*}
which is a weighted $\ell_{\alpha}$-minimization problem. Particularly, the case $\alpha=2$ corresponds to a least squares problem which can be solved efficiently. The third step updates the smoothing parameter $\varepsilon$ according to the following rule \footnote{Various strategies can be applied to update $\varepsilon$. For example, we can keep $\varepsilon$ as a small fixed value. It is preferred to choose a sequence of $\{\varepsilon^{(k)}\}$ tending to zero \citep{Daubechiesabc2010}.}
\[ \varepsilon^{(k)} = \min\{\varepsilon^{(k-1)}, \rho \cdot r(\D{\bet}^{(k)})_{l}\}~\textrm{ with } \rho\in(0,1),\]
where $r(\D{\bet})_l$ is the $l$-th smallest element of the set $\{|\boldsymbol\D_{j.}{\bet}| : j=1,\ldots,n\}$. ${\bet}$ is a $l$-cosparse vector if and only if $r(\D \bet)_l=0$. The algorithm stops when $\varepsilon=0$.

\begin{algorithm}[!htb]
   \caption{The CoIRLq Algorithm}
   \label{alg:2}
\begin{algorithmic}
   \STATE {\bfseries Input:} $l$, ${\X}$, $\y$, $\D=[\boldsymbol\D_{1.}; \ldots ; \boldsymbol\D_{n.}]$.
   \STATE {\bfseries Init:} Choose ${\bet}^{(0)}$ such that ${\X}{\bet}^{(0)}=\y$ and $\varepsilon^{(0)}=1$.

   \WHILE{ $\|{{\bet}}^{(k{+}1)}-{{\bet}}^{(k)}\|_{\infty} >\tau$ or $\varepsilon^{(k)}\neq0$ }
   \STATE Update
     \[
      \eta_i^{(k)} = \Big( |\boldsymbol\D_{i.} {\bet}^{(k-1)}|^{\alpha} + ({\varepsilon^{(k-1)}})^{\alpha} \Big)^{\frac{q}{\alpha}-1}, i=1,\ldots,n .
      \]

   \STATE Update
  \begin{equation*}
    {{\bet}}^{(k)} = \argmin_{{\bet}\in \mathbb{R}^d} \Big\{\frac{1}{2}||\y-{\X}{\bet}||_2^2 + \frac{\lambda q}{\alpha} \sum_{i=1}^n \eta_i^{(k)}|\boldsymbol\D_{i.}{\bet}|^{\alpha} \Big\}.
  \end{equation*}

   \STATE Update
   \[ \varepsilon^{(k)} = \min\{\varepsilon^{(k-1)}, \rho \cdot r(\D{\bet}^{(k)})_{l}\}\textrm{ with } \rho\in(0,1).\]

   \ENDWHILE
   \STATE {\bfseries Output:} ${{\bet}}$
\end{algorithmic}
\end{algorithm}

\subsection{Convergence Analysis}

Our analysis is based on the optimization problem (\ref{coirls}) with the objective function $F({\bet},\varepsilon)$.
 Noting that $\boldsymbol{\eta}^{(k+1)}$ is a function of ${\bet}^{(k)}$ and $\varepsilon^{(k)}$, we define the following objective function
\begin{align*}
& Q({\bet},\varepsilon|{\bet}^{(k)},\varepsilon^{(k)}) \triangleq \frac{1}{2}||\y-{\X}{\bet}||_2^2   \\
&~~~~ +\frac{\lambda q}{\alpha} \sum_{i=1}^n  \Big[ \eta_i^{(k+1)}(|\boldsymbol\D_{i.}{\bet}|^{\alpha}+\varepsilon^{\alpha}) + \frac{\alpha-q}{q} \frac{1}{{\eta_i^{(k+1)}}^{\frac{q}{\alpha-q}}} \Big].
\end{align*}

\begin{lemma} \label{lemma2}
Assume that the analysis operator $\D$ has full column rank. Let $\{({\bet}^{(k)}, \varepsilon^{(k)}): k=1, 2, \ldots \}$ be a sequence generated by the CoIRLq algorithm. Then,
\[ ||{\bet}^{(k)}||_2 \leq  \sigma_{min}^{-1}(\D) (F({\bet}^{(0)},\varepsilon^{(0)})/\lambda)^{1/q} \]
and
\[ F({\bet}^{(k+1)},\varepsilon^{(k+1)}) \leq F({\bet}^{(k)},\varepsilon^{(k)}), \]
with equality holding if and only if ${\bet}^{(k+1)}={\bet}^{(k)}$ and $\varepsilon^{(k+1)}=\varepsilon^{(k)}$.
\end{lemma}

The boundedness of $||{\bet}^{(k)}||_2$ implies that the sequence $\{{\bet}^{(k)}\}$ converges to some accumulation point.
We can immediately derive the convergence property of the CoIRLq algorithm  from Zangwill's \emph{global convergence theorem} or the literature \citep{SriperumbudurLanckriet2009}. Here we omit the detail. 
Finally, it is easy to verify that when $\varepsilon^{*}=0$, $\bet^{*}$ is a stationary point of (\ref{unlq}).

\subsection{Recovery Guarantee Analysis}

To uniquely recover the true parameter, the linear operator $\X : \mathcal{A}\rightarrow\mathbb{R}^m$ must be a one-to-one map. Define a set $\bar{\mathcal{A}}=\{{\bet}={\bet}_1+{\bet}_2 : {\bet}_1,{\bet}_2\in\mathcal{A}\}$. \citet{BlumensathDavies2008} showed that a necessary condition for the existence of a one-to-one map requires that $\delta_{\bar{\mathcal{A}}}<1$ $(\delta_{\mathcal{A}}\leq\delta_{\bar{\mathcal{A}}})$. For any two $l$-cosparse vectors ${\bet}_1, {\bet}_2\in\mathcal{A}=\{{\bet}: \D_{\Lambda}{\bet}=\boldsymbol{0},|\Lambda|\geq l\}$, denote $T_1=supp(\D{\bet}_1)$, $T_2=supp(\D{\bet}_2)$, $\Lambda_1=cosupp(\D{\bet}_1)$ and $\Lambda_2=cosupp(\D{\bet}_2)$. Since $supp(\D({\bet}_1+{\bet}_2))\subseteq T_1\cup T_2$, we have $cosupp(\D({\bet}_1+{\bet}_2))\supseteq (T_1\cup T_2)^c = T_1^c\cap T_2^c = \Lambda_1\cap \Lambda_2$. Moreover, we also have $|\Lambda_1\cap \Lambda_2|=n-|T_1\cup T_2|\geq n-(n-l)-(n-l)=2l-n$.
Thus it requires that the linear operator $\X$ satisfies the $\mathcal{A}$-RIP with $\delta_{2l-n}<1$ to uniquely recover any $l$-cosparse vector from the set $\mathcal{A}=\{{\bet}:\D_{\Lambda}{\bet}=0,|\Lambda|\geq l\}$. Otherwise, there would exist two $l$-cosparse vectors ${\bet}_1\neq{\bet}_2$ such that $\X({\bet}_1-{\bet}_2)=0$. \citet{Giryesabc2013}  showed that there  exists a random matrix $\X$ satisfying such a requirement with high probability.
\begin{theorem} \label{thm:2}
Let ${\bet}^{*}\in\mathbb{R}^d$ be a $l$-cosparse vector, and $\y={\X}{\bet}^{*}+\w$ with $||\w||_2\leq\epsilon$. Assume that ${\X}$ satisfies the $\cal A$-RIP over the set $\mathcal{A}=\{{\bet}:\D_{\Lambda}{\bet}=\boldsymbol{0},|\Lambda|\geq l\}$ of order $2l-n$ with $\delta_{2l-n}<1$. Then the solution $\hat{{\bet}}$ obtained by the CoIRLq algorithm obeys
   \[ ||\hat{\bet} - \bet^{*}||_2 \leq C_1\sqrt{\lambda} + C_2 \epsilon, \]
where $C_1$ and $C_2$ are constants depending on $\delta_{2l-n}$.
\end{theorem}
We can see that the CoIRLq algorithm can recover an approximate solution away from the true parameter vector by a factor of $\sqrt{\lambda}$ in the noiseless case. 

\section{Numerical Analysis}

In this section we conduct  numerical analysis of the $\ell_q$-analysis minimization method on both simulated data and real data, and compare the performance of the case $q<1$ and the case $q=1$. We set $\alpha=2$ in the CoIRLq algorithm. 

\subsection{Cosparse Vector Recovery}

We generate the simulated datasets according to 
\[ \y={\X}{\bet}+\w, \]
where $\w \sim N(\boldsymbol{0},\sigma\boldsymbol{I})$. The sampling matrix ${\X}$ is drawn independently from the normal distribution with normalized columns. The analysis operator $\D$ is constructed such that $\D^T$is a random tight frame. To generate a $l$-cosparse vector ${\bet}$, we first choose $l$ rows randomly from $\D$ and form $\D_{\Lambda}$.Then we generate a vector which lies in the null space of $\D_{\Lambda}$. The recovery is deemed to be successful if the recovery relative error $||\hat{{\bet}}-{\bet}^{*}||_2/||{\bet}^{*}||_2 \leq 1e-4$.

In the first experiment, we test the vector recovery capability of the CoIRLq method with $q=0.7$. We set $m=80,n=144,d=120,l=99,$ and $\sigma=0$. Figure 1 illustrates that the CoIRLq method recovers the original vector perfectly.
\begin{figure}[htb]
\centering
{
\includegraphics[width=0.52\linewidth]{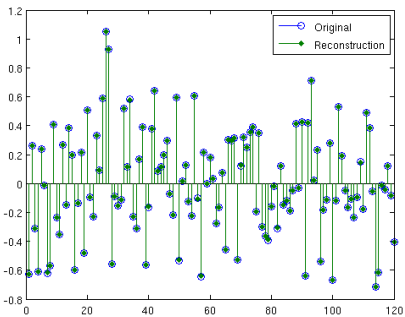}  \\
}
\caption{Cosparse vector recovery. }
\label{fig:1}
\end{figure}

In the second experiment, we test the CoIRLq method on a range of sample size and cosparsity with different $q$. Although the optimal tuning parameter $\lambda$ depends on $q$, a small enough $\lambda$ is able to ensure that $\y$ approximately equals to ${\X}{\bet}$ in the noiseless case. Thus, we set $\lambda=1e-4$ for all $q$ and $\sigma=0$. Figure 2 reports the result with 100 repetitions on every dataset. We can see that the CoIRLq method with $q=0.5, 0.7, 0.8$ can achieve exact recovery in a wider range of cosparsity and with fewer samples than with $q=1$. In addition, it should be noted that  small $q=0.1$ or $q=0.3$ do not perform better than  relatively large $q=0.7, 0.8$, because a too small $q$ leads to a hard-solving problem. Note that there is a drop of recovery probability where the cosparsity $l=118$ \footnote{When $l=120$, a zero vector is generated by our codes. So the recovery probability in cosparsity $l=120$ is zero.}. This is because  it is hard to algorithmically recover a vector residing in a subspace with a small dimension; please also refer to \citep{Namabc2011}.
\begin{figure}[htb]
\center
\scriptsize
\center
    \begin{tabular}{c@{}c}
        \includegraphics[width=0.45\linewidth]{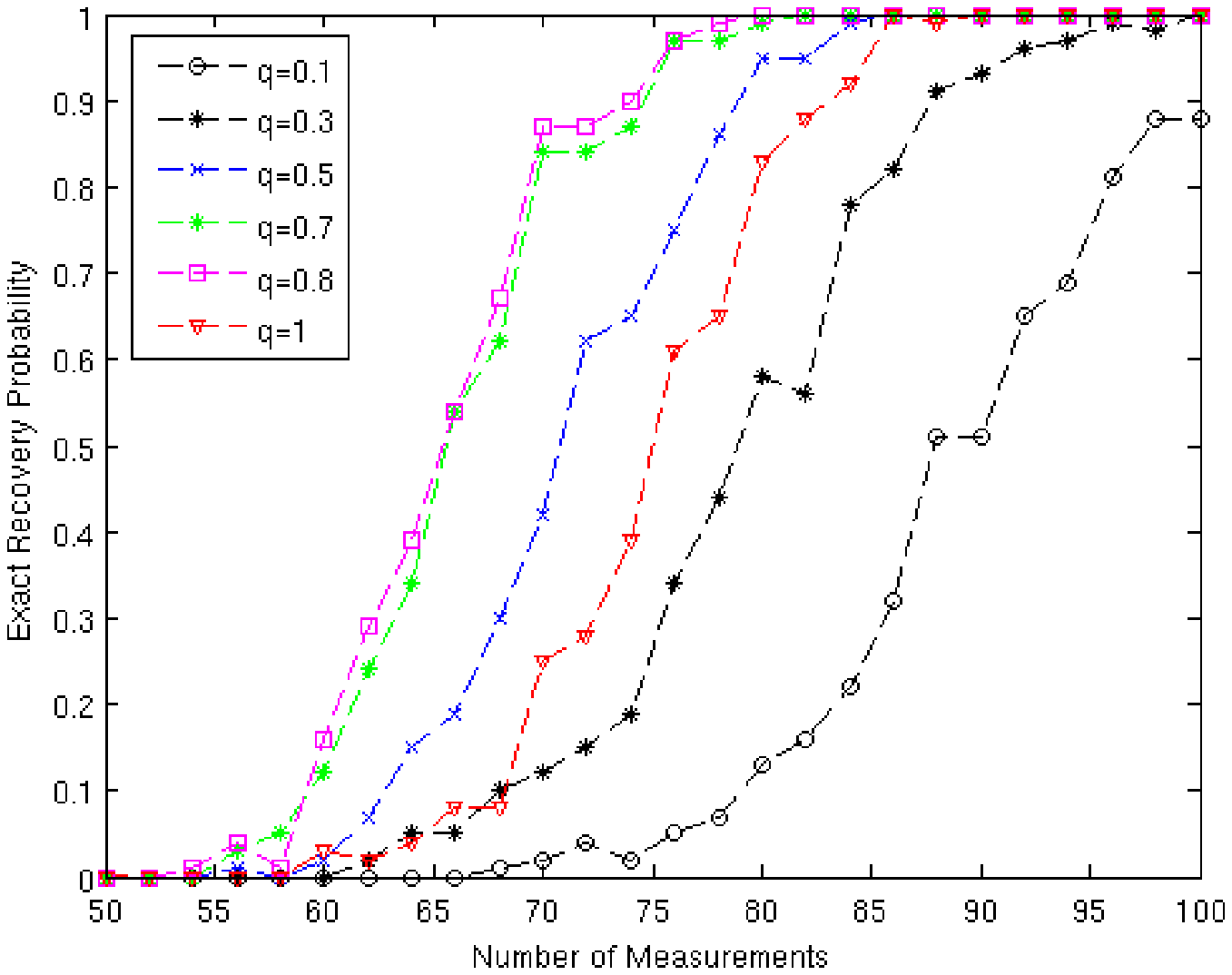}
        ~&~\includegraphics[width=0.45\linewidth]{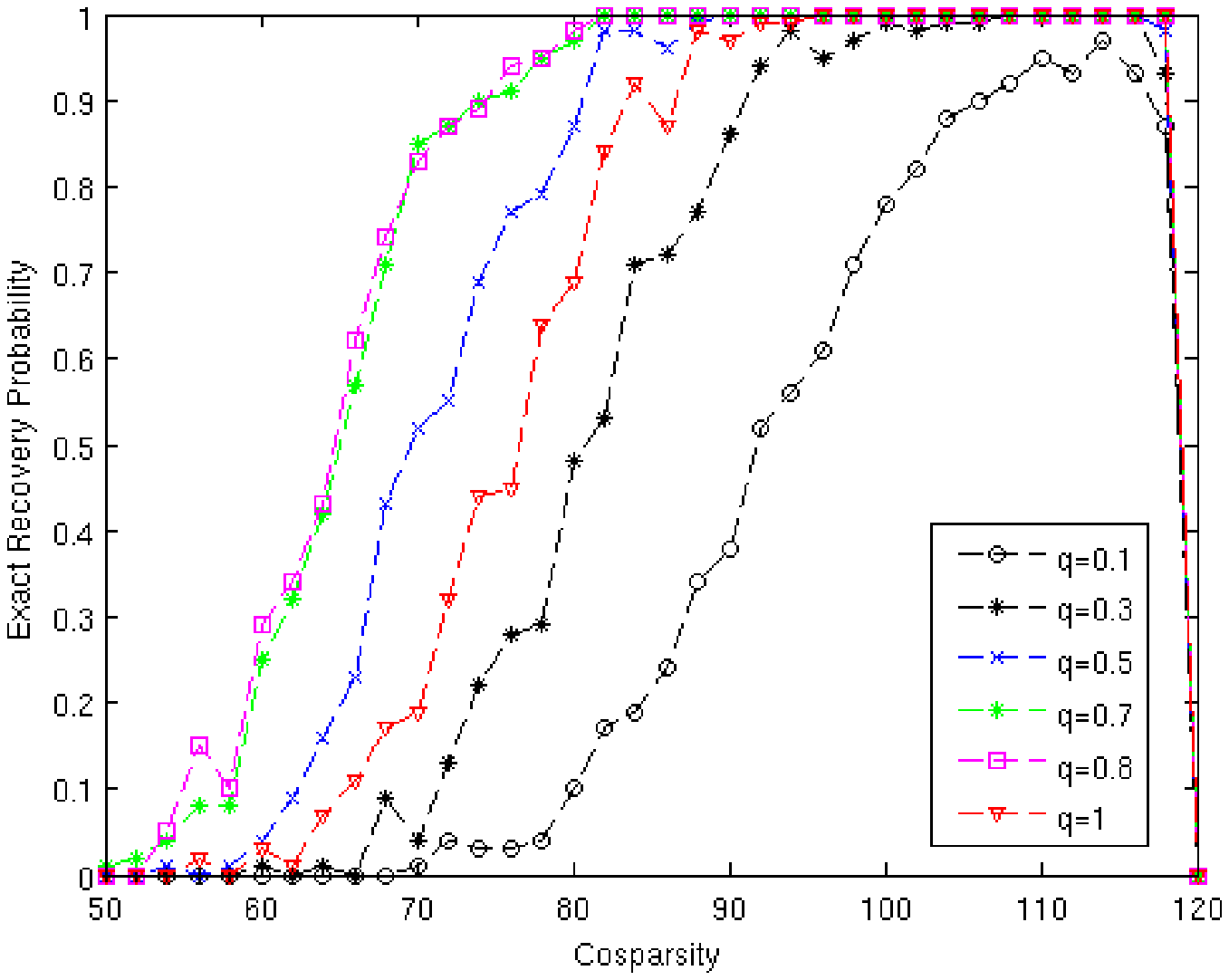}\\
        $n=144$, $d=120$, $l=99$ ~&~ $m=90$, $n=144$, $d=120$ 
    \end{tabular}
\caption{Exact recovery probability of the CoIRLq method.}
\label{fig:2}
\end{figure}
\begin{figure}[htb]
\center
\scriptsize
\center
    \begin{tabular}{c@{}c}
        \includegraphics[width=0.45\linewidth]{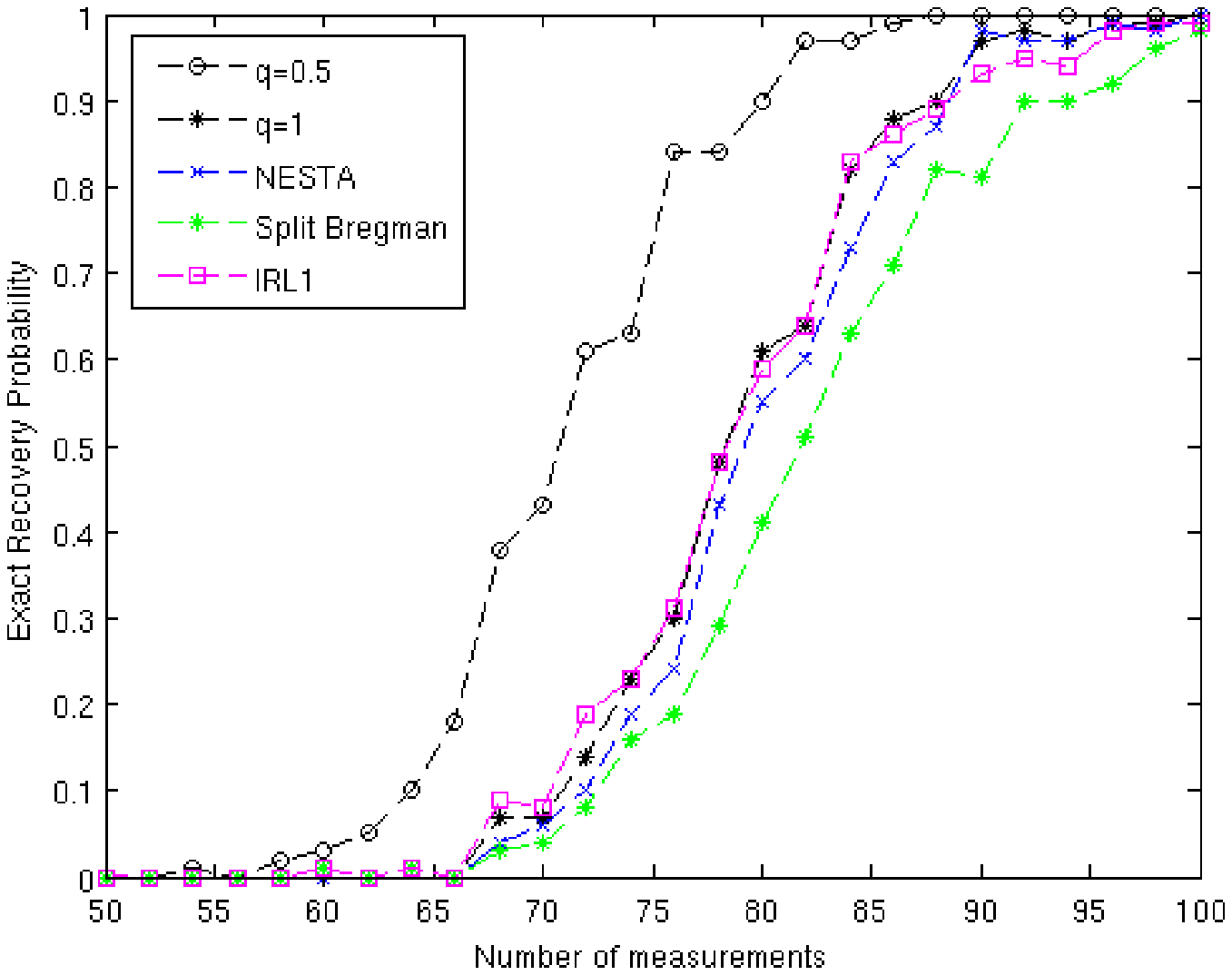}
        ~&~\includegraphics[width=0.45\linewidth]{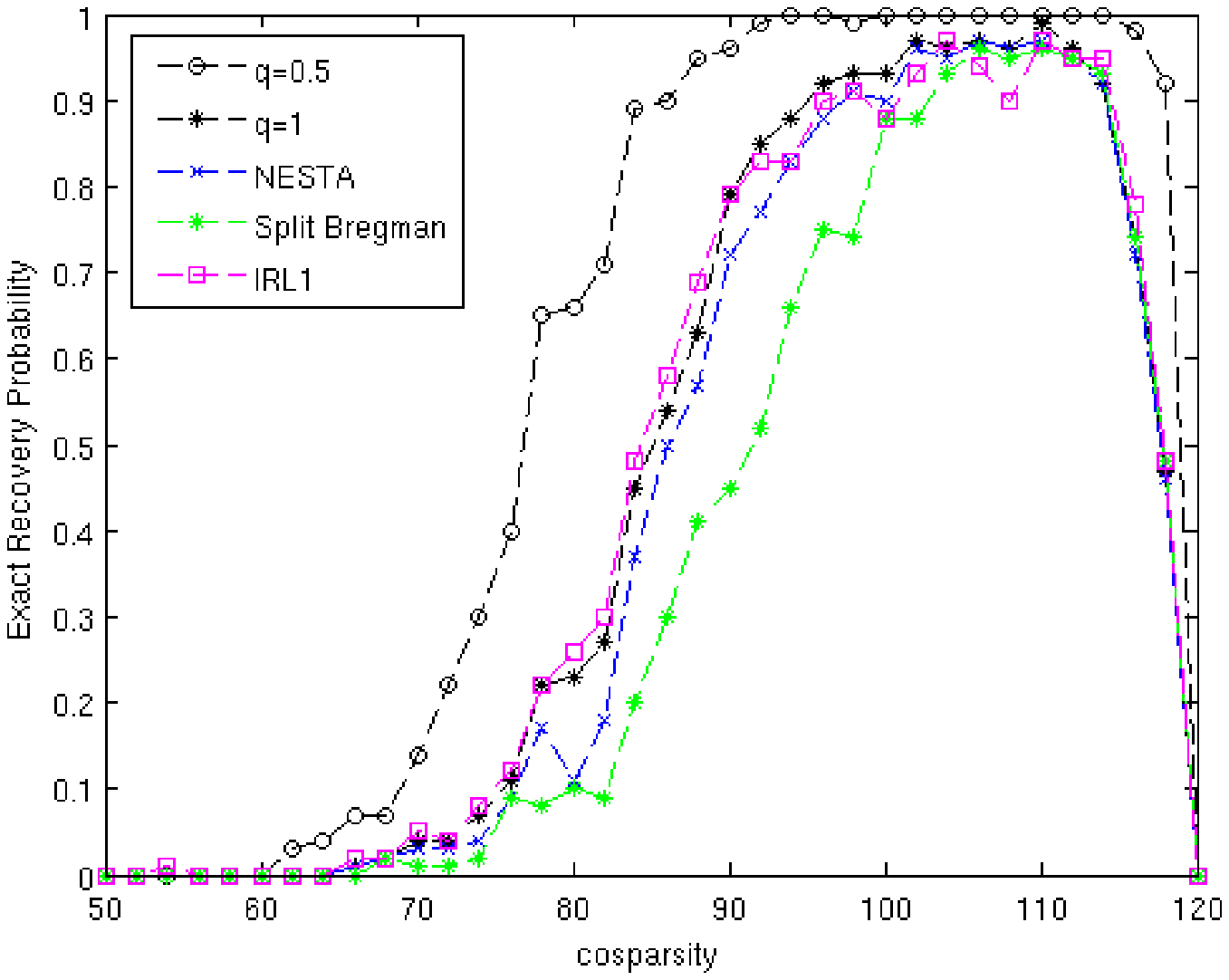}\\
        {$n=144$, $d=120$, $l=99$} ~&~ {$m=90$, $n=144$, $d=120$}
    \end{tabular}
\caption{Recovery probability of the CoIRLq, NESTA, IRL1 and split Bregman methods.}
\label{fig:3}
\end{figure}

In the third experiment, we compare the CoIRLq method with three state-of-the-art methods for the $\ell_1$-analysis minimization problem including NESTA ({http://statweb.stanford.edu/$\sim$candes/nesta/}), split Bregman method, and iteratively reweighted $\ell_1$ (IRL1) method. Set the noise level $\sigma=0.01$. The parameter $\lambda$ is tuned via the grid search method. We run these methods in a range of sample size and cosparsity. Figure 3 reports the result with 100 repetitions on every dataset. We can see that the nonconvex $\ell_q$-analysis minimization with $q<1$ is more capable of achieving exact recovery against noise than the convex $\ell_1$-analysis minimization. Moreover, the nonconvex approach can obtain exact recovery with fewer samples or in a wider range of cosparsity than the convex counterpart. Moreover, we found that the CoIRLq algorithm in the case $q<1$ often needs less iterations than in the case $q=1$. 

\subsection{Image Restoration Experiment}

\begin{figure}[!htb]
\begin{minipage}{0.2\linewidth}
  \centerline{\includegraphics[width=4cm]{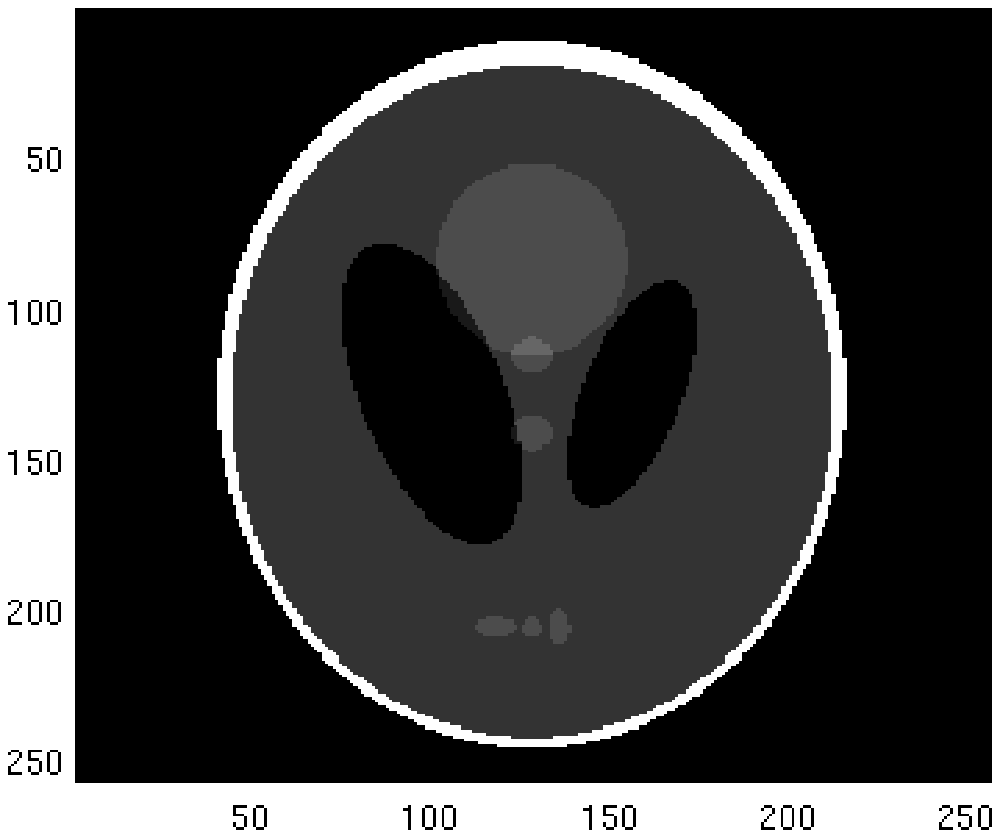}}
  \centerline{(a) Original image} 
\end{minipage}
\hfill
\begin{minipage}{0.2\linewidth}
  \centerline{\includegraphics[width=4cm]{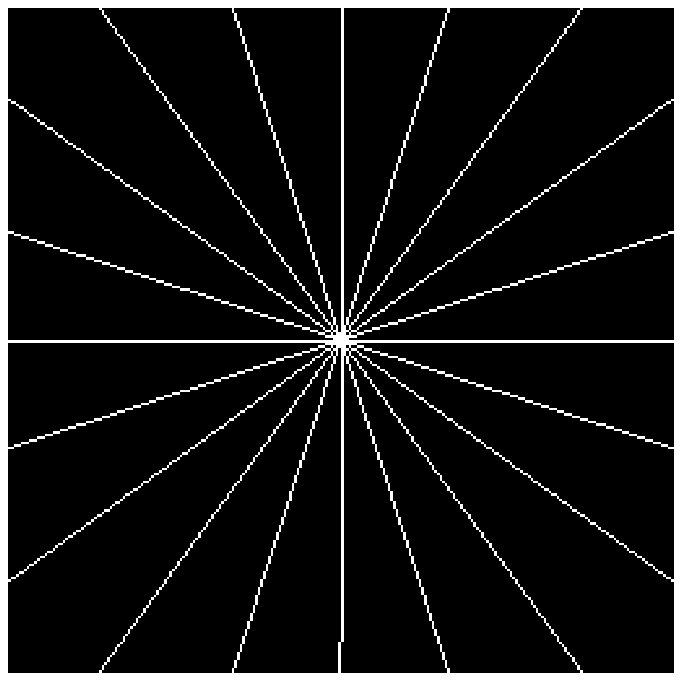}}
  \centerline{(b) 10 lines}
\end{minipage}
\hfill
\begin{minipage}{0.2\linewidth}
  \centerline{\includegraphics[width=4cm]{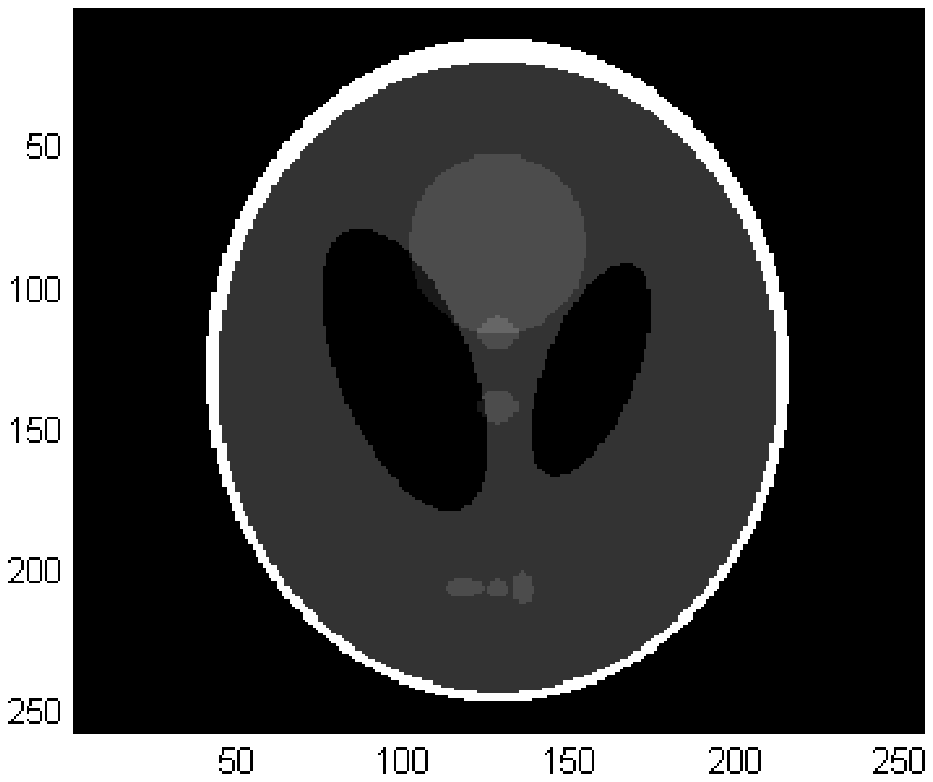}}
  \centerline{(c) SNR=107.7}
\end{minipage}
\hfill
\begin{minipage}{0.2\linewidth}
  \centerline{\includegraphics[width=4cm]{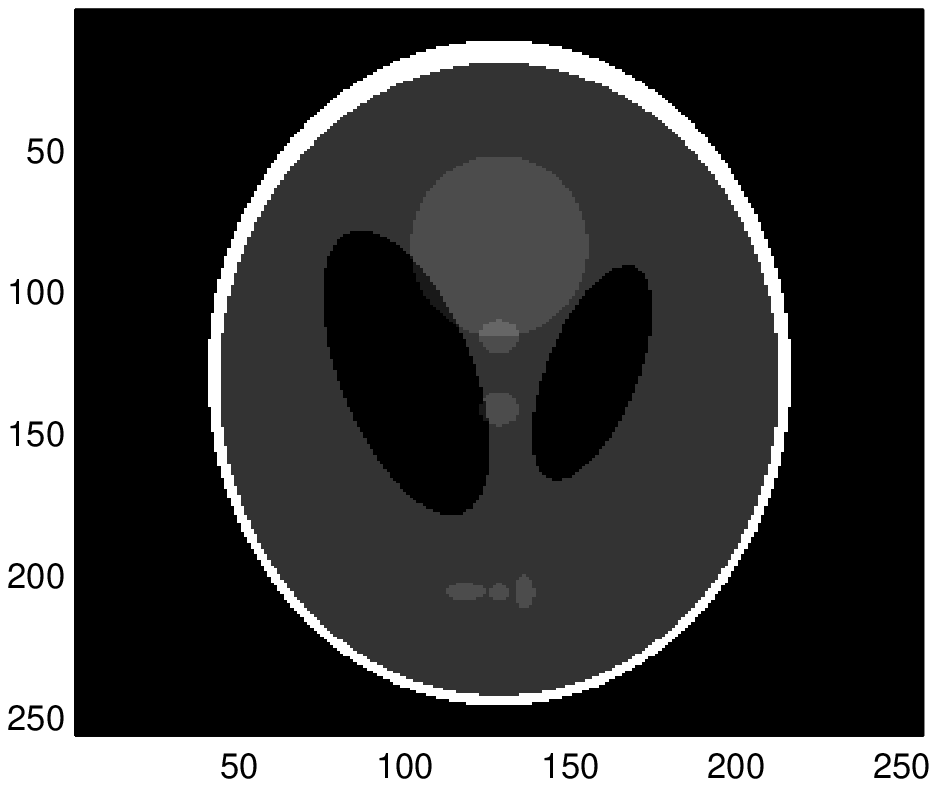}}
  \centerline{(d) SNR=83.6}
\end{minipage}
\vfill
\begin{minipage}{0.2\linewidth}
  \centerline{\includegraphics[width=4cm]{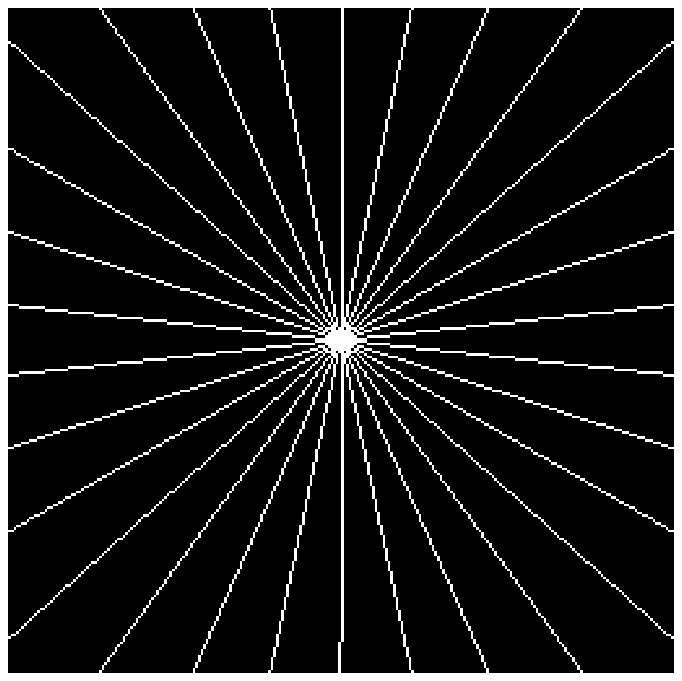}}
  \centerline{(e) 15 lines} 
\end{minipage}
\hfill
\begin{minipage}{0.2\linewidth}
  \centerline{\includegraphics[width=4cm]{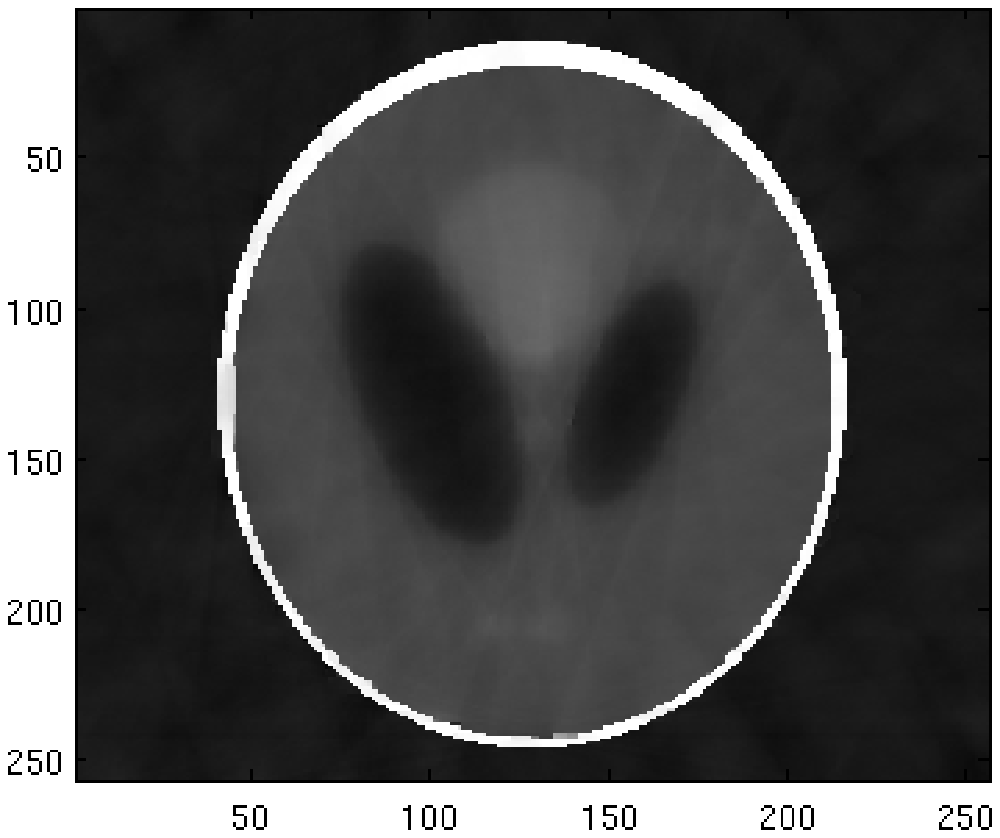}}
  \centerline{(f) SNR=30.5}
\end{minipage}
\hfill
\begin{minipage}{0.2\linewidth}
  \centerline{\includegraphics[width=4cm]{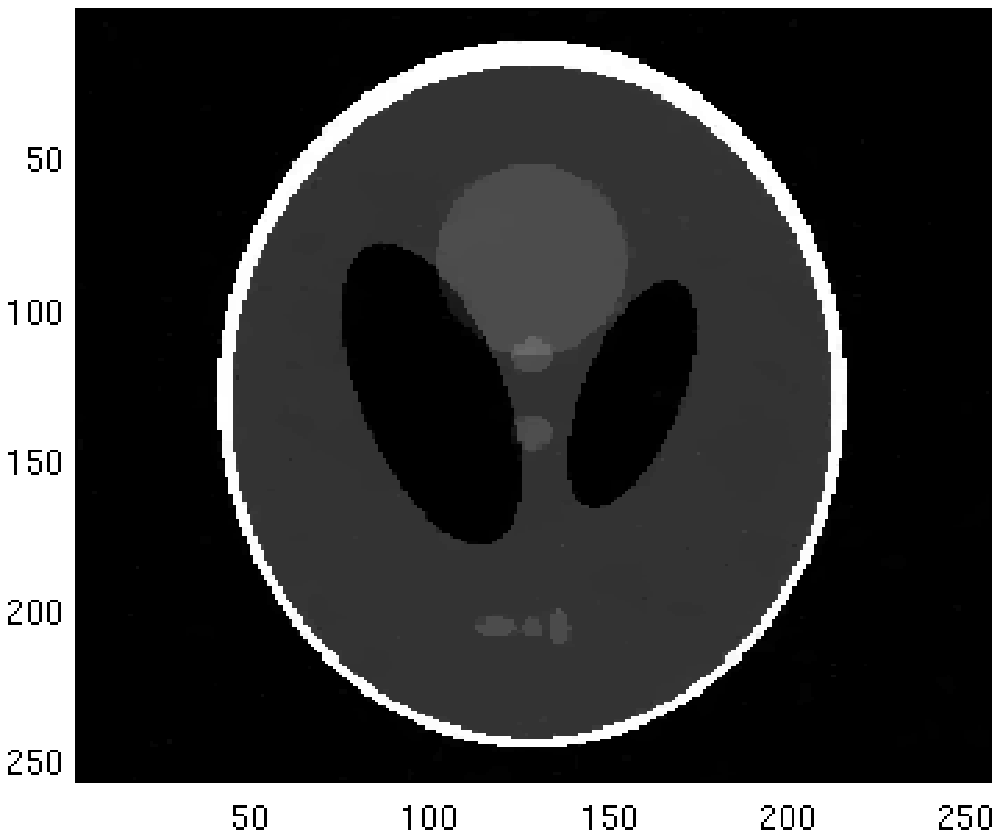}}
  \centerline{(g) SNR=45.7}
\end{minipage}
\hfill
\begin{minipage}{0.2\linewidth}
  \centerline{\includegraphics[width=4cm]{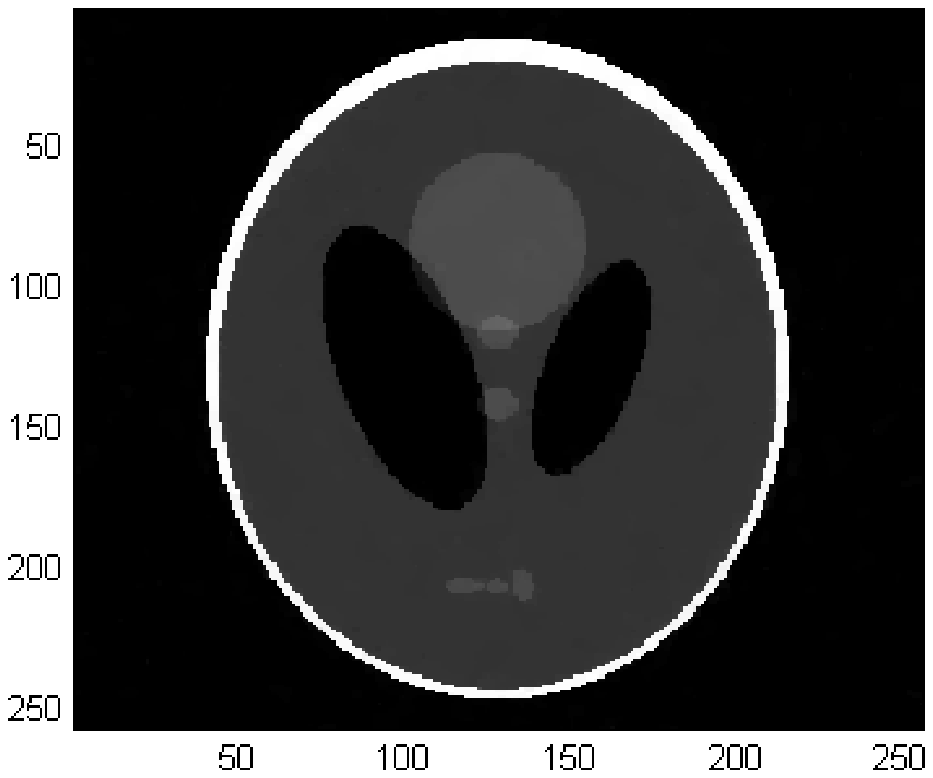}}
  \centerline{(h) SNR=43}
\end{minipage}
\caption{(a) Original Shepp Logan Phantom image; (b) Sampling locations along 10 radial lines; (c) Exact reconstruction via CoIRLq (q=0.7) with 10 lines without noise; (d) Exact reconstruction via CoIRLq (q=1) with 12 lines without noise; (e) Sampling locations along 15 radial lines; (f) Reconstruction via GAP with 15 lines and noise level $\sigma=0.01$; (g) Reconstruction via CoIRLq (q=0.7) with 15 lines and noise level $\sigma=0.01$; (h) Reconstruction via CoIRLq (q=1) with 15 lines and noise level $\sigma=0.01$.}
\label{fig:4}
\end{figure}

In this section we demonstrate the effectiveness of the $\ell_q$-analysis minimization on the Shepp Logan phantom reconstruction problem. In computed tomography, an image can not be observed directly. Instead, we can only obtain its 2D Fourier transform coefficients along a few radial lines due to certain limitations. This sampling process can be modeled as a measurement matrix $\X$. The goal is to reconstruct the image from the observation.

The experimental program is set as follows. The image dimension is of $256\times 256$, namely $d=65536$. The measurement matrix $\X$ is a two dimensional Fourier transform which measures the image's Fourier transform along a few radial lines. The analysis operator is a finite difference operator $\D_{\textrm{2D-DIF}}$ whose size is roughly twice the image size, namely $n=130560$. Since the number of nonzero analysis coefficients is $n-l=2546$, the cosparsity used is $l=n-2546=128014$. The number of measurements depends on the number of radial lines used. To show the reconstruction capability of the CoIRLq method, we conduct the following experiments (the parameter $\lambda$ is tuned via grid search). First, we compare our method with the greedy analysis pursuit (GAP~{http://www.small-project.eu/software-data}) method for the $\ell_0$-analysis minimization.

Figures~\ref{fig:4}-(f), (g) and (h) show that our method performs better than the GAP method in the noisy case. We can see that the CoIRLq method with $q<1$ is more robust to noise than the case with $q=1$. Second, we take an experiment using 10 radial lines without noise. The corresponding number of measurements is $m=2282$, which is approximately $3.48\%$ of the image size. Figure~\ref{fig:4}-(c) demonstrates that the CoIRLq ($q=0.7$) method with 10 lines obtains perfect reconstruction. Figure~\ref{fig:4}-(d) shows that the CoIRLq ($q=1$) method with 12 lines attains perfect reconstruction. However, the GAP method  needs at least 12 radial lines to achieve exact recovery; see \citep{Namabc2011}.

\section{Conclusion}

In this paper we have conducted the  theoretical analysis and developed the computational method, for the $\ell_q$-analysis minimization problem. Theoretically, we have established weaker conditions for exact recover in noiseless case and a tighter non-asymptotic upper bound of estimate error in noisy case. In particular, we have presented a necessary and sufficient condition guaranteeing exact recovery. 
Additionally, we have shown that the nonconvex $\ell_q$-analysis optimization can do recovery with a lower sample complexity and in a wider range of cosparsity. Computationally, we have devised an iteratively reweighted method to solve the $\ell_q$-analysis optimization problem. 
Empirical results have illustrated that our iteratively reweighted method outperforms the state-of-the-art methods.

\bibliographystyle{plainnat}
\bibliography{icml15}

\end{document}


\title{Appendix}
\date{}

\maketitle

Consider the following $\ell_q$-analysis minimization ($0<q\leq 1$) problem
\begin{align} \label{analysis_lq}
 \min_{\boldsymbol{\beta}} ||\boldsymbol{D}\boldsymbol{\beta}||_q^q ~~~~\textrm{ s.t. }~~~~ ||\boldsymbol{y}-\boldsymbol{X}\boldsymbol{\beta}||_2 \leq \epsilon, 
\end{align}
where $\bet\in\mathbb{R}^d$ is the parameter vector of interest, $\X\in\mathbb{R}^{m\times d}$ is a matrix with far fewer rows than columns, $\y\in\mathbb{R}^m$ is the observation, and $\epsilon$ is the upper bound of noise vector.
                                                                                                                                                                                                                                                                                                                                                                  
\textbf{Notations} Let $\boldsymbol{h}=\hat{\boldsymbol{\beta}} - \boldsymbol{\beta}^{*}$ ($\hat{\boldsymbol{\beta}}$ denotes the solution of (\ref{analysis_lq}), $\boldsymbol{\beta}^{*}$ denotes the population parameter). $|\boldsymbol{v}|_{(k)}$ denotes the $k$-th element of $|\boldsymbol{v}|=[|v_1|,|v_2|,\ldots]^T$. The submatrix $\D_{T}$ is constructed by replacing the rows of $\D$ corresponding to $T^c$ by zero rows. Denote $\D_T\bet=(\D\bet)_T$. $\D^{+}$ is the pseudoinverse of $\D$. $\sigma_{max}(\boldsymbol{D})$ and $\sigma_{min}(\boldsymbol{D})$ denote respectively the maximal and minimal nonzero singular value of $\boldsymbol{D}$. If $\D^{+}$ has full column rank, we have $\D^{+}\D=\I$ and $\sigma_{min}(\boldsymbol{D})=\sigma_{max}^{-1}(\boldsymbol{D}^{+})$, $\sigma_{max}(\boldsymbol{D})=\sigma_{min}^{-1}(\boldsymbol{D}^{+})$. Define $\kappa = \frac{\sigma_{max}(\boldsymbol{D})}{\sigma_{min}(\boldsymbol{D})} = \frac{\sigma_{max}(\boldsymbol{D}^{+})}{\sigma_{min}(\boldsymbol{D}^{+})}$. Null($\D$) denotes the null space of $\D$. ~$\lfloor \cdot\rfloor$ and $\lceil \cdot\rceil$ denotes the rounding down and up operator respectively. The best $k$-term approximation error of $\bet\in\mathbb{R}^d$ with $\ell_q$ norm is defined as $\sigma_{k}(\boldsymbol{\beta})_q := \inf_{\z\in\{\z\in\mathbb{R}^d: ||\z||_0\leq k \}}  ||\boldsymbol{\beta}-\boldsymbol{z}||_q $. Let $\boldsymbol{e}_i\in\mathbb{R}^d$ be the indicator vector with only one nonzero entry of $1$ or $-1$ in position $i$.

\begin{definition}
 \textbf{($\mathcal{A}$-restricted isometry property)} A matrix $\boldsymbol{\Phi}\in\mathbb{R}^{m\times d}$ obeys the $\mathcal{A}$-restricted isometry property with constant $\delta_{\mathcal{A}}$ over any subset $\mathcal{A}\in\mathbb{R}^d$, if $\delta_{\mathcal{A}}$ is the smallest quantity satisfying
\[ (1-\delta_{\mathcal{A}})||{\v}||_2^2 \leq ||\boldsymbol{\Phi}{\v}||_2^2 \leq (1+\delta_{\mathcal{A}})||{\v}||_2^2 \]
for all ${\v}\in\mathcal{A}$.
\end{definition}
When choosing $\mathcal{A}=\{{\D\v}: \v\in\mathbb{R}^d, ||\v||_0\leq k\}$ and $\mathcal{A}=\{{\v}: \v\in\mathbb{R}^d, \D_{\Lambda}\v = \boldsymbol{0}, |\Lambda|\geq l\}$, the $\mathcal{A}$-restricted isometry are the D-RIP and $\Omega$-RIP respectively.

\begin{definition}
 A matrix $\boldsymbol{\Phi}\in\mathbb{R}^{m\times d}$ satisfies the Null Space Property adapted to $\D$ (D-NSPq) of order $k$  if
\begin{equation} \label{D-NSPq}
  \forall \boldsymbol{v}\in \textrm{Null} (\boldsymbol{\Phi})/\{\boldsymbol{0}\}, \forall |T|\leq k,  ||\boldsymbol{D}_T\boldsymbol{v}||_q^q < ||\boldsymbol{D}_{T^c}\boldsymbol{v}||_q^q .
\end{equation}
\end{definition}

\section{The Proof of Theorem~1}

\begin{lemma} \label{io_gnsp} (Theorem 1, [1])
Given matrices $\boldsymbol{X}\in\mathbb{R}^{m\times d}$,$\boldsymbol{D}\in\mathbb{R}^{n\times d}$,  finite union of low-dimensional linear subspaces (UoS) denoted by $\Sigma$, and norms $||\cdot||_X, ||\cdot||_Y$, the following are equivalent: \\
I. Instance optimality - there exists some decoder $\bigtriangleup$ solving the inverse problem $\boldsymbol{y}=\boldsymbol{X}\boldsymbol{\beta}$ such that
\begin{equation} \label{io}
 ||\bigtriangleup(\boldsymbol{X}\boldsymbol{\beta}) - \boldsymbol{D}\boldsymbol{\beta}||_X \leq C_1 d(\boldsymbol{D}\boldsymbol{\beta},\boldsymbol{D}\Sigma)_Y, \forall \boldsymbol{\beta}\in\mathbb{R}^d,
\end{equation}
II. Generalized Null Space Property:
\begin{equation} \label{gnsp}
 ||\boldsymbol{D v}||_X \leq  C_2 d(\boldsymbol{D v},\boldsymbol{D} (\Sigma+\Sigma))_Y, \forall \boldsymbol{v}\in\textrm{Null}(\boldsymbol{X}) ,
\end{equation}
where the constants $C_1,C_2$ are related by a factor at most $2$, $\Sigma+\Sigma=\{ \boldsymbol{z}: \boldsymbol{z}=\boldsymbol{z}_1+\boldsymbol{z}_2, \boldsymbol{z}_1,\boldsymbol{z}_2\in\Sigma \}$, and $d(\boldsymbol{\beta},\Sigma)_Y$ is the distance (in the sense of the $Y$ norm) from $\boldsymbol{\beta}$ to $\Sigma$.
\end{lemma}
In fact, the above lemma holds for any metric. Thus it also holds with respect to $||\cdot||_q^q$ for $0<q\leq 1$.

\begin{theorem} \label{exact_nsp}
Let $\bet\in\mathbb{R}^d$ with cosupport $\Lambda$, $||\D\bet||_0=k$, and $\y=\X\bet$. Then $\bet$ is the unique minimizer of the $\ell_q$-analysis minimization (\ref{analysis_lq}) with $\epsilon=0$ if and only if $\X$ satisfies the D-NSPq (\ref{D-NSPq}) relative to $\Lambda^c$.
\end{theorem}
\begin{proof}
First we prove the neccessity. Given a fixed cosupport $\Lambda$, define $\Sigma :=\{{\u}\in\mathbb{R}^d: ||\u||_0 \leq \frac{|\Lambda^c|}{2} \}$. Define the decoder $\bigtriangleup(\boldsymbol{y})=\boldsymbol{D}\bigtriangleup_0(\boldsymbol{y})$ with $\bigtriangleup_0(\boldsymbol{y}) =\argmin_{\boldsymbol{y}=\boldsymbol{X}\boldsymbol{\beta}} ||\boldsymbol{D}\boldsymbol{\beta}||_q^q$. Assume that every cosparse vector $\boldsymbol{\beta}$ with cosupport $\Lambda$ is the unique minimizer of $||\boldsymbol{D}\boldsymbol{z}||_q^q$ subject to $\boldsymbol{y}=\boldsymbol{X}\boldsymbol{z}$. Then the instance optimality (\ref{io}) in Lemma \ref{io_gnsp} holds since $||\bigtriangleup(\boldsymbol{X}\boldsymbol{\beta}) - \boldsymbol{D}\boldsymbol{\beta}||_X=0$. Thus  we have $||\boldsymbol{D v}||_X \leq  C_2 d(\boldsymbol{D v},\boldsymbol{D}(\Sigma+\Sigma))_Y, \forall \boldsymbol{v}\in\textrm{Null}(\boldsymbol{X})$. Set $||\cdot||_X$ as $||\cdot||_q^q$ and $d(\boldsymbol{a},\boldsymbol{b})_Y=\inf_{\boldsymbol{b}\in\D(\Sigma+\Sigma)} ||\boldsymbol{a}-\boldsymbol{b}||_q^q$. Then we get $||\boldsymbol{D}\boldsymbol{v}||_q^q \leq  C_1 ||\boldsymbol{D}_{\Lambda}\boldsymbol{v}||_q^q, \forall \boldsymbol{v}\in\textrm{Null}(\boldsymbol{X})$. This establishes the null space property $||\boldsymbol{D}_{\Lambda^c}\boldsymbol{v}||_q^q \leq  (C_2-1) ||\boldsymbol{D}_{\Lambda}\boldsymbol{v}||_q^q < ||\boldsymbol{D}_{\Lambda}\boldsymbol{v}||_q^q$ for $1<C_2<2$.

Now we prove the sufficiency. Given a vector $\boldsymbol{\beta}$ with cosupport $\Lambda$ ($\boldsymbol{D}_{\Lambda}\boldsymbol{\beta}=0$) and another vector $\boldsymbol{z}$ satisfying $\y=\boldsymbol{X}\bet=\boldsymbol{X}\z$, we consider $\boldsymbol{v}:=\boldsymbol{\beta}-\boldsymbol{z}\in \textrm{Null}(\boldsymbol{X})/\{\boldsymbol{0}\}$. Assuming the null space property adapted to $\boldsymbol{D}$ holds, then we have
\begin{align*}
& ||\boldsymbol{D}\boldsymbol{\beta}||_q^q \leq ||\boldsymbol{D}\boldsymbol{\beta}-\boldsymbol{D}_{\Lambda^c}\boldsymbol{z}||_q^q + ||\boldsymbol{D}_{\Lambda^c}\boldsymbol{z}||_q^q = ||\boldsymbol{D}_{\Lambda^c}\boldsymbol{v}||_q^q + ||\boldsymbol{D}_{\Lambda^c}\boldsymbol{z}||_q^q \\
&~~~~~~~~~~~~~~~~~~~~~~~~~~~~~ < ||\boldsymbol{D}_{\Lambda}\boldsymbol{v}||_q^q + ||\boldsymbol{D}_{\Lambda^c}\boldsymbol{z}||_q^q =||\boldsymbol{D}_{\Lambda}\boldsymbol{z}||_q^q + ||\boldsymbol{D}_{\Lambda^c}\boldsymbol{z}||_q^q = ||\boldsymbol{D}\boldsymbol{z}||_q^q.
\end{align*}
The first inequality is because of that $||\cdot||_q^q$ satisifes the triangle inequality. It verifies that $\boldsymbol{\beta}$ is the unique minimizer.

\end{proof}

\begin{corollary}
 For $0<q_1 <q_2\leq 1$, the sufficient condition for exact recovery via the $\ell_{q_2}$-analysis minimization is also sufficient for exact recovery via the $\ell_{q_1}$-analysis minimization.
\end{corollary}
\begin{proof}
 Given an index set $T$ ($|T|=k$), suppose that the D-NSPq (\ref{D-NSPq}) relative to $T$  holds for the $\ell_q$-analysis minimization ($0<q\leq 1$). Let $\boldsymbol{D v}= \sum_{i=1}^n a_i\boldsymbol{e}_i$, $\v\in\textrm{Null}(\X)/\{\boldsymbol{0}\}$ and $T$ is the index set of $k$ largest absolute entries of $\D\v$. The D-NSPq (\ref{D-NSPq}) can be rewritten as
 \[
  \sum_{j\in T} |a_j|^q  < \sum_{l\in T^c} |a_l|^q.
 \]
It can be further reformulated as
\[
  \sum_{j\in T} \frac{1}{\sum_{l\in T^c} (|a_l|/|a_j|)^q }  < 1.
\]
Since $|a_l|/|a_j|<1$ for $j\in T$ and $l\in T^c$, the left-hand side of the above inequality is nondecreasing with respect to $q\leq1$. Thus $\textrm{D-NSP}_{q_2}$ implies $\textrm{D-NSP}_{q_1}$ for $q_1<q_2\leq1$.

\end{proof}

\section{The Proof of Theorem~2}

\begin{lemma} \label{sparsepolytope} (Lemma 1.1, [2])
 (Sparse Representation of a Polytope) For a positive number $\alpha$ and a positive integer $s$, define the polytope $T(\alpha,s)\subset\mathbb{R}^n$ by 
 \[
  T(\alpha,s)\subset\mathbb{R}^n = \{ \v\in\mathbb{R}^n : ||\v||_{\infty}\leq \alpha, ||\v||_1\leq s\alpha \}.
 \]
For any $\v\in\mathbb{R}^n$, define the set of sparse vectors $U(\alpha,s,\v)\subset \mathbb{R}^n$ by
\[
 U(\alpha,s,\v) = \{ \u\in\mathbb{R}^n : \textrm{supp}(\u)\subseteq \textrm{supp}(\v), ||\u||_0\leq s, ||\u||_1=||\v||_1, ||\u||_{\infty} \leq \alpha \}.
\]
Then $\v\in T(\alpha,s)$ if and only if $\v$ is in the convex hull of $U(\alpha,s,\v)$. In particular, any $\v\in T(\alpha,s)$ can be expressed as 
\[
 \v = \sum_{i=1}^N \lambda_i \u_i, \textrm{ and }, 0\leq \lambda_i \leq 1, \sum_{i=1}^N \lambda_i = 1, \textrm{ and } \u_i\in U(\alpha,s,\v).
\]
\end{lemma}

\begin{theorem} \label{exact_rip}
Let $\bet\in\mathbb{R}^d$, $||\D\bet||_0=k$, and $\y=\X\bet$.  Assume that $\D\in\mathbb{R}^{n\times d}$ has full column rank, and its condition number is upper bounded by $\kappa< \sqrt{\frac{2\rho+1+\sqrt{4\rho+1}}{2\rho}}$. If $\X\in\mathbb{R}^{m\times d}$ satisfies the $\mathcal{A}$-RIP over the set $\mathcal{A}=\{{\D\v}: ||\v||_0 \leq (t^q+1)k\}$ with $k,t^q k\in\mathbb{N}, t>0,q\in(0,1]$, i.e.,
\begin{equation}\label{noiseless_cond}
\delta_{(t^q+1)k} < \frac{\rho(1-\kappa^4) + \kappa^2 \sqrt{4\rho+1}}  { \rho (\kappa^2+1)^2 + \kappa^2 }
\end{equation}
with $\rho= t^{q-2}/4$, then $\bet$ is the unique minimizer of the $\ell_q$-analysis minimization
 (\ref{analysis_lq}) with $\epsilon=0$. 
\end{theorem}
\begin{proof}
Let $T$ deonte the support of $\D\bet$ ($|T|=k$). By the Null Space Property (\ref{D-NSPq}), we only need to verify for all $\boldsymbol{h}\in\textrm{Null}(\boldsymbol{X})/\{\boldsymbol{0}\}$, $||\boldsymbol{D}_{T} \boldsymbol{h}||_q^q < ||\boldsymbol{D}_{T^c} \boldsymbol{h}||_q^q$. Suppose that there exists $\boldsymbol{h}\in\textrm{Null}(\boldsymbol{X})/\{\boldsymbol{0}\}$ such that $||\boldsymbol{D}_{T} \boldsymbol{h}||_q^q \geq ||\boldsymbol{D}_{T^c} \boldsymbol{h}||_q^q$. Without loss of generality, assume that $(\boldsymbol{Dh})_1\geq  \ldots \geq (\boldsymbol{Dh})_n \geq 0$. 
Set $\alpha^q = ||(\boldsymbol{D}_{T} \boldsymbol{h})^q||_1 / k$. Divide $\boldsymbol{D}_{T^c} \boldsymbol{h}$ into two parts, i.e., $\boldsymbol{D}_{T^c} \boldsymbol{h} = \boldsymbol{h}_1 + \boldsymbol{h}_2$ where
\begin{align*}
 & \boldsymbol{h}_1 = \boldsymbol{D}_{T^c} \boldsymbol{h} \cdot \boldsymbol{1}_{ \{i : (\boldsymbol{D}_{T^c} \boldsymbol{h})_i > \alpha/t \} }, \\
 & \boldsymbol{h}_2 = \boldsymbol{D}_{T^c} \boldsymbol{h} \cdot \boldsymbol{1}_{ \{i : (\boldsymbol{D}_{T^c} \boldsymbol{h})_i \leq \alpha/t \} }. 
\end{align*}
Then $(\boldsymbol{D}_{T^c} \boldsymbol{h})^q = \boldsymbol{h}_1^q+\boldsymbol{h}_2^q$. And $||\boldsymbol{h}_1^q||_1 \leq ||(\boldsymbol{D}_{T^c} \boldsymbol{h})^q||_1\leq k\alpha^q$. Since all non-zero entries of $\boldsymbol{h}_1^q$ has magnitude larger than $(\alpha/t)^q$, we have
\[
 k\alpha^q \geq ||\boldsymbol{h}_1^q||_1 \geq \sum_{i\in\textrm{supp}(\boldsymbol{h}_1^q)} (\alpha/t)^q = m(\alpha/t)^q,
\]
where $m=|\textrm{supp}(\boldsymbol{h}_1^q)|$. So the sparsity of $\boldsymbol{h}_1^q$ is at most $t^q k$ (assume $t^qk$ is an integer). Thus we have
\begin{align*}
 & ||\boldsymbol{h}_2^q||_1 = ||(\boldsymbol{D}_{T^c}\boldsymbol{h})^q||_1 - ||\boldsymbol{h}_1^q||_1 \leq k\alpha^k - \frac{m\alpha^q}{t^q} = (t^qk-m)(\frac{\alpha}{t})^q, \\
 & ||\boldsymbol{h}_2^q||_{\infty} \leq (\frac{\alpha}{t})^q.
\end{align*}
According to Lemma \ref{sparsepolytope}, $\boldsymbol{h}_2^q$ lies in a polytope and can be expressed as a convex combination of sparse vectors, i.e., $\boldsymbol{h}_2^q= \sum_{i=1}^N \lambda_i \boldsymbol{u}_i$ where $\boldsymbol{u}_i$ is a $(t^q k-m)$-sparse vector. Suppose that $\mu\geq 0, c\geq 0$ are to be determined. Let $\boldsymbol{\theta}_i^q= (\boldsymbol{D}_{T} \boldsymbol{h})^q + \boldsymbol{h}_1^q + \mu\boldsymbol{u}_i$, then
\begin{equation} \label{equ1}
\begin{align*}
 \sum_{j=1}^N \lambda_j \boldsymbol{\theta}_j^q - c\boldsymbol{\theta}_i^q &= (\boldsymbol{D}_{T} \boldsymbol{h})^q + \boldsymbol{h}_1^q + \mu\boldsymbol{h}_2^q - c\boldsymbol{\theta}_i^q  \\
 &= (1-\mu-c)( (\boldsymbol{D}_{T} \boldsymbol{h})^q + \boldsymbol{h}_1^q) - c\mu\boldsymbol{u}_i + \mu (\boldsymbol{D}\boldsymbol{h})^q.
\end{align*}
\end{equation}
Since $(\boldsymbol{D}_{T} \boldsymbol{h})^q, \boldsymbol{h}_1^q, \boldsymbol{u}_i$ are $k$-, $m$-, and $(t^q k-m)$-sparse respectively, then $\boldsymbol{\theta}_i^q, \sum_{j=1}^N \lambda_j \boldsymbol{\theta}_j^q - c\boldsymbol{\theta}_i^q - \mu (\boldsymbol{D}\boldsymbol{h})^q$ are all $(t^q+1)k$-sparse vectors.

Define $\boldsymbol{\Lambda} := \textrm{diag}((\boldsymbol{Dh})_1^{1-q},(\boldsymbol{Dh})_2^{1-q},\ldots,(\boldsymbol{Dh})_n^{1-q}), \boldsymbol{B}:=\boldsymbol{X}\boldsymbol{D}^{+}\boldsymbol{\Lambda}$. Then $\boldsymbol{B}(\boldsymbol{Dh})^q=\boldsymbol{X}\boldsymbol{D}^{+}\boldsymbol{\Lambda}(\boldsymbol{Dh})^q=\boldsymbol{X}\boldsymbol{D}^{+}\boldsymbol{D}\boldsymbol{h}=\boldsymbol{Xh}=\boldsymbol{0}$. We also get
\begin{align*}
 ||\boldsymbol{\Lambda}\boldsymbol{u}_i||_2^2 &\leq \sum_{j=k+1+m}^{(t^q+1)k}  (|(\boldsymbol{Dh})_j|^{1-q} ||\boldsymbol{u}_i||_{\infty} )^2 \\
 &\leq (t^qk-m)((\frac{\alpha}{t})^{1-q} (\frac{\alpha}{t})^{q} )^2  \\
 &\leq t^{q-2} k \alpha^2 \\
 &= t^{q-2} k \left( \frac{||\boldsymbol{D}_{T} \boldsymbol{h}||_q}{k^{1/q}} \right)^2  \\
 &\leq t^{q-2} k \left( \frac{k^{1/q-1/2}||\boldsymbol{D}_{T} \boldsymbol{h}||_2}{k^{1/q}} \right)^2 \\
 &= t^{q-2} ||\boldsymbol{D}_{T} \boldsymbol{h}||_2^2 \\
 &\leq t^{q-2} ||\boldsymbol{D}_{T} \boldsymbol{h} + \boldsymbol{h}_1||_2^2 .
\end{align*}
Additionally, we have
\begin{equation}
\begin{align*}
& \boldsymbol{B}(\sum_{j=1}^N \lambda_j\boldsymbol{\theta}_j^q-c\boldsymbol{\theta}_i^q) =  \boldsymbol{B} [(1-\mu-c)( (\boldsymbol{D}_{T} \boldsymbol{h})^q + \boldsymbol{h}_1^q) - c\mu\boldsymbol{u}_i + \mu (\boldsymbol{D}\boldsymbol{h})^q]  \\
&~~~~~~~~~~~ = \boldsymbol{X}\boldsymbol{D}^{+}\boldsymbol{\Lambda} [(1-\mu-c)( (\boldsymbol{D}_{T} \boldsymbol{h})^q + \boldsymbol{h}_1^q) - c\mu\boldsymbol{u}_i + \mu (\boldsymbol{D}\boldsymbol{h})^q] \\
&~~~~~~~~~~~ = \boldsymbol{X}\boldsymbol{D}^{+} [(1-\mu-c)( \boldsymbol{D}_{T} \boldsymbol{h} + \boldsymbol{h}_1) - c\mu \boldsymbol{\Lambda}\boldsymbol{u}_i + \mu \boldsymbol{D}\boldsymbol{h}] \\
&~~~~~~~~~~~ = \boldsymbol{X}\boldsymbol{D}^{+} [(1-\mu-c)( \boldsymbol{D}_{T} \boldsymbol{h} + \boldsymbol{h}_1) - c\mu \boldsymbol{\Lambda}\boldsymbol{u}_i ] \\
\\
&\boldsymbol{B}\boldsymbol{\theta}_i^q = \boldsymbol{X}\boldsymbol{D}^{+}\boldsymbol{\Lambda} [(\boldsymbol{D}_{T} \boldsymbol{h})^q + \boldsymbol{h}_1^q + \mu\boldsymbol{u}_i]  \\
&~~~~~~~~~~~ = \boldsymbol{X}\boldsymbol{D}^{+} ( \boldsymbol{D}_{T} \boldsymbol{h} + \boldsymbol{h}_1 + \mu \boldsymbol{\Lambda} \boldsymbol{u}_i).
\end{align*}
\end{equation}
It is easy to check the following identity
\begin{equation} \label{l2_identity}
\begin{align*}
& \sum_{i=1}^N \lambda_i ||\boldsymbol{B}(\sum_{j=1}^N \lambda_j\boldsymbol{\theta}_j^q-c\boldsymbol{\theta}_i^q)||_2^2 + (1-2c) \sum_{1\leq i\leq j\leq N} \lambda_i\lambda_j ||\boldsymbol{B}(\boldsymbol{\theta}_i^q-\boldsymbol{\theta}_j^q)||_2^2 \\
&~~~~~~~~~~~~~~~~~~~~~~~ = \sum_{i=1}^N \lambda_i (1-c)^2 ||\boldsymbol{B}\boldsymbol{\theta}_i^q||_2^2.
\end{align*}
\end{equation}

Set $c=1/2$, write $\delta$ for $\delta_{(t^q+1)k}$, let the right-hand side of the equation (\ref{l2_identity}) minus the left-hand side, and obtain
\begin{equation} \label{RIP_inequ2}
\begin{align*} 
 0 &= \sum_{i=1}^N \lambda_i ||\boldsymbol{B}(\sum_{j=1}^N \lambda_j\boldsymbol{\theta}_j^q-c\boldsymbol{\theta}_i^q)||_2^2 - \sum_{i=1}^N \lambda_i (1-c)^2 ||\boldsymbol{B}\boldsymbol{\theta}_i^q||_2^2 \\
 &= \sum_{i=1}^N \lambda_i ||\boldsymbol{X}\boldsymbol{D}^{+} [(1-\mu-c)( \boldsymbol{D}_{T} \boldsymbol{h} + \boldsymbol{h}_1) - c\mu \boldsymbol{\Lambda}\boldsymbol{u}_i ] ||_2^2   \\
 &~~~~~~~~~~~~~~~~~~~~~ - \sum_{i=1}^N \lambda_i (1-c)^2 ||\boldsymbol{X}\boldsymbol{D}^{+} ( \boldsymbol{D}_{T} \boldsymbol{h} + \boldsymbol{h}_1 + \mu \boldsymbol{\Lambda} \boldsymbol{u}_i)||_2^2  \\
 &\leq (1+\delta) \sum_{i=1}^N \lambda_i ||\boldsymbol{D}^{+} [ (\frac{1}{2}-\mu)( \boldsymbol{D}_{T} \boldsymbol{h} + \boldsymbol{h}_1) - \frac{\mu}{2} \boldsymbol{\Lambda}\boldsymbol{u}_i  ] ||_2^2   \\
 &~~~~~~~~~~~~~~~~~~~~~ - (1-\delta) \sum_{i=1}^N \lambda_i \frac{1}{4}  ||\boldsymbol{D}^{+}  (\boldsymbol{D}_{T} \boldsymbol{h} + \boldsymbol{h}_1   + \mu \boldsymbol{\Lambda} \boldsymbol{u}_i)  ||_2^2   \\
 &\leq (1+\delta) \sum_{i=1}^N \lambda_i  \sigma_{max}^2(\boldsymbol{D}^{+})  ||(\frac{1}{2}-\mu)( \boldsymbol{D}_{T} \boldsymbol{h} + \boldsymbol{h}_1) - \frac{\mu}{2} \boldsymbol{\Lambda}\boldsymbol{u}_i  ||_2^2   \\
 &~~~~~~~~~~~~~~~~~~~~~ - (1-\delta) \sum_{i=1}^N \lambda_i \frac{1}{4} \sigma_{min}^2(\boldsymbol{D}^{+}) ||\boldsymbol{D}_{T} \boldsymbol{h} + \boldsymbol{h}_1   + \mu \boldsymbol{\Lambda} \boldsymbol{u}_i  ||_2^2   \\
 &\leq (1+\delta) \kappa^2 [ (\frac{1}{2}-\mu)^2 || \boldsymbol{D}_{T} \boldsymbol{h} + \boldsymbol{h}_1||_2^2 + \frac{\mu^2}{4} \sum_{i=1}^N \lambda_i ||\boldsymbol{\Lambda}\boldsymbol{u}_i  ||_2^2  ]  \\
 &~~~~~~~~~~~~~~~~~~~~~ - \frac{1-\delta}{4}  (  ||\boldsymbol{D}_{T} \boldsymbol{h} + \boldsymbol{h}_1 ||_2^2  + \mu^2 \sum_{i=1}^N \lambda_i || \boldsymbol{\Lambda} \boldsymbol{u}_i  ||_2^2 )  \\
 &=  (\frac{1}{2}-\mu)^2 (1+\delta) \kappa^2 || \boldsymbol{D}_{T} \boldsymbol{h} + \boldsymbol{h}_1||_2^2 - \frac{1}{4}(1-\delta)  || \boldsymbol{D}_{T} \boldsymbol{h} + \boldsymbol{h}_1||_2^2     \\ 
 &~~~~~~~~~~~~~~~~~~~~~  + (\frac{1}{4} \mu^2(1+\delta) \kappa^2 - \frac{1}{4} \mu^2(1-\delta) ) \sum_{i=1}^N \lambda_i  || \boldsymbol{\Lambda}\boldsymbol{u}_i ||_2^2  \\
 &\leq  (\frac{1}{2}-\mu)^2 (1+\delta) \kappa^2 || \boldsymbol{D}_{T} \boldsymbol{h} + \boldsymbol{h}_1||_2^2 - \frac{1}{4}(1-\delta)  || \boldsymbol{D}_{T} \boldsymbol{h} + \boldsymbol{h}_1||_2^2     \\ 
 &~~~~~~~~~~~~~~~~~~~~~  + (\frac{1}{4} \mu^2(1+\delta)\kappa^2 - \frac{1}{4} \mu^2(1-\delta) )  t^{q-2} ||\boldsymbol{D}_T\boldsymbol{h}+\boldsymbol{h}_1||_2^2  \\
 &=  \left( (\frac{1}{2}-\mu)^2 (1+\delta) \kappa^2 - \frac{1}{4}(1-\delta)  + \frac{1}{4}  t^{q-2}\mu^2(\kappa^2(1+\delta)-(1-\delta)) \right) || \boldsymbol{D}_{T} \boldsymbol{h} + \boldsymbol{h}_1||_2^2    .  
\end{align*}
\end{equation}
If 
\begin{equation}\label{cond1_noiseless}
\delta_{(t^q+1)k} < \frac{ \rho (1-\kappa^4) + \kappa^2 \sqrt{4\rho+1}} { \rho(\kappa^2+1)^2 + \kappa^2 }
\end{equation}
with $\rho= t^{q-2}/4$, then there exists some $\mu$ such that the following holds
\[
  (\frac{1}{2}-\mu)^2 (1+\delta) \kappa^2 - \frac{1}{4}(1-\delta)  + \frac{1}{4}  t^{q-2}\mu^2(\kappa^2(1+\delta)-(1-\delta)) < 0.
\]
And it leads contradiction, i.e., $0<0$.
Therefore, the D-RIP condition (\ref{cond1_noiseless}) is sufficient to guarantee the D-NSPq (\ref{D-NSPq}) holding. Additionally, we have 
\begin{equation}\label{cond2_noiseless}
 \delta_{(t^q+1)k} < \frac{ (\rho - \rho\kappa^2 - \kappa^2)\mu^2 + \kappa^2\mu + \frac{1}{4}(1-\kappa^2) }{ (\rho + \rho\kappa^2 + \kappa^2)\mu^2 -\kappa^2\mu + \frac{1}{4} (1+\kappa^2 ) }  .
\end{equation}
To determin $\mu$, the right-hand sides of (\ref{cond1_noiseless}) and (\ref{cond2_noiseless}) should equal. Thus we obtain
\begin{equation}\label{mu}
 \mu = \frac{(b+1)\kappa^2 + \sqrt{ - \rho(b+1)^2\kappa^4 + (2\rho+1)(1-b^2)\kappa^2  -\rho(1-b)^2 }} { 2(\rho+1)(b+1)\kappa^2 + 2\rho (b-1) },
\end{equation}
where $\rho= t^{q-2}/4, b = \frac{\rho (1-\kappa^4) + \kappa^2 \sqrt{4\rho+1}} { \rho(\kappa^2+1)^2 + \kappa^2 }$. When $t^qk$ is not an integer, set ${t^{'}}^q = \lceil t^q k\rceil / k$. It can be verified that $b(t)$ is monotonically increasing with respect to $t$. Thus we have
\[
 \delta_{({t^{'}}^q+1)k} = \delta_{(t^q+1)k} < b(t) < b(t^{'}).
\]

The right-hand side of the inequality (\ref{cond1_noiseless}) should be great than zero, i.e.,
\[
 \rho (1-\kappa^4)+\kappa^2\sqrt{4\rho+1} >0.
\]
Thus we get 
\[
 \kappa < \sqrt{\frac{2\rho+1+\sqrt{4\rho+1}}{2\rho}}.
\]

\end{proof}

\section{The Proof of Theorem~3}

\begin{lemma} \label{norm_inequ} (Lemma 5.3, [3])
 Suppose $a_1\geq a_2 \geq \cdots \geq a_m \geq 0, \lambda\geq 0$ and $\sum_{i=1}^r a_i + \lambda \geq \sum_{j=r+1}^m a_i$, then for all $\alpha \geq 1$,
 \[
  \left( \sum_{j=r+1}^m a_i^{\alpha} \right)^{\frac{1}{\alpha}}  \leq  \left( \sum_{i=1}^r a_i^{\alpha} \right)^{\frac{1}{\alpha}} + \frac{\lambda}{r^{1-1/\alpha}}.
 \]

\end{lemma}

\begin{lemma} \label{cone_constraint}
 Let $\boldsymbol{h}=\hat{\boldsymbol{\beta}} - \boldsymbol{\beta}^{*} $ and $T$ deonte the set of indices corresponding to the $k$ largest entries of $\D\bet^{*}$. The following inequality holds
 \[ 
 ||\boldsymbol{D}_{T^c} \boldsymbol{h}||_q^q \leq ||\boldsymbol{D}_{T} \boldsymbol{h}||_q^q + 2||\boldsymbol{D}_{T^c} \boldsymbol{\beta}^{*}||_q^q  ,
 \]
 and
 \[
   ||\boldsymbol{D}_{T^c} \boldsymbol{h}||_2^2 \leq  ||\boldsymbol{D}_{T} \boldsymbol{h}||_2^2 + \frac{ 2^{2/q}  ||\boldsymbol{D}_{T^c} \boldsymbol{\beta}^{*}||_q^{2}  }{k^{2/q-1} }  .
 \]

\end{lemma}
\begin{proof}
Both $\boldsymbol{\beta}^{*}$ and $\hat{\boldsymbol{\beta}}$ are the feasible solutions. But $\hat{\boldsymbol{\beta}}$ is the minimizer of the problem (\ref{analysis_lq}). Thus we have
\begin{align*} 
||\boldsymbol{D} \boldsymbol{\beta}^{*}||_q^q \geq ||\boldsymbol{D} \hat{\boldsymbol{\beta}}||_q^q &= ||\boldsymbol{D} \boldsymbol{\beta}^{*}-\boldsymbol{D}\boldsymbol{h}||_q^q = ||\boldsymbol{D}_{T} \boldsymbol{\beta}^{*} -\boldsymbol{D}_{T} \boldsymbol{h}||_q^q + ||\boldsymbol{D}_{T^c} \boldsymbol{\beta}^{*} -\boldsymbol{D}_{T^c} \boldsymbol{h}||_q^q \\
 &\geq ||\boldsymbol{D}_{T}\boldsymbol{\beta}^{*}||_q^q -||\boldsymbol{D}_{T}\boldsymbol{h}||_q^q  +||\boldsymbol{D}_{T^c}\boldsymbol{h}||_q^q - ||\boldsymbol{D}_{T^c}\boldsymbol{\beta}^{*}||_q^q 
\end{align*}
The second inequality is because of that the pseudonorm $||\cdot||_q^q$ satisfies the triangle inequality.

According to Lemma \ref{norm_inequ}, we set $a_i = (\boldsymbol{Dh})_i^q$ and $\alpha = 2/q$. Then we get
\[
  ( ||\boldsymbol{D}_{T^c} \boldsymbol{h}||_2^2)^{q/2} \leq (||\boldsymbol{D}_{T} \boldsymbol{h}||_2^2)^{q/2} + \frac{ 2||\boldsymbol{D}_{T^c} \boldsymbol{\beta}^{*}||_q^q }{k^{1-q/2} }  .
 \]
\end{proof}

\begin{lemma} \label{tube_constraint}
 The following inequality holds
\[ ||\boldsymbol{X h}||_2\leq 2\epsilon . \]
\end{lemma}
\begin{proof}
 Since both $\boldsymbol{\beta}^{*}$ and $\hat{\boldsymbol{\beta}}$ are feasible, we have
\begin{align*}
 ||\boldsymbol{X h}||_2 = ||\boldsymbol{y}-\boldsymbol{X}\boldsymbol{\beta}^{*} -(\boldsymbol{y}-\boldsymbol{X}\hat{\boldsymbol{\beta}})||_2 \leq ||\boldsymbol{y}-\boldsymbol{X}\boldsymbol{\beta}^{*}||_2 + ||\boldsymbol{y}-\boldsymbol{X}\hat{\boldsymbol{\beta}}||_2 \leq 2\epsilon.
\end{align*}

\end{proof}

\begin{lemma} \label{tail_constraint}
Assume that $\D\in\mathbb{R}^{n\times d}$ has full column rank, and its condition number is upper bounded by $\kappa< \sqrt{\frac{2\rho+1+\sqrt{4\rho+1}}{2\rho}}$. If $\X\in\mathbb{R}^{m\times d}$ satisfies the $\mathcal{A}$-RIP over the set $\mathcal{A}=\{{\D\v}: ||\v||_0 \leq (t^q+1)k \}$ with $k, t^q k\in\mathbb{N}, t>0, q\in(0,1] $, i.e.,
\begin{equation} \label{noisy_cond}
\delta_{(t^q+1)k} < \frac{\rho(1-\kappa^4) + \kappa^2 \sqrt{4\rho+1}} { \rho (\kappa^2+1)^2 + \kappa^2 }
\end{equation}
with $\rho= 4^{1/q-2}t^{q-2}$, then the following holds
\begin{align*}
  ||\boldsymbol{D}_{T} \boldsymbol{h}||_2 \leq c_1 \epsilon + c_2 \frac{2^{1/q} ||\boldsymbol{D}_{T^c}\boldsymbol{\beta}^{*}||_q}{k^{1/q-1/2}},
\end{align*}
where
\begin{align*}
 & c_0 = (\frac{1}{2}-\mu)^2 (1+\delta) \kappa^2 - \frac{1}{4}(1-\delta_{(t^q+1)k})   
  + \rho \mu^2(\kappa^2(1+\delta_{(t^q+1)k})-(1-\delta_{(t^q+1)k})),  \\
 & c_1 = \frac{ 2\kappa(\mu-\mu^2)\sqrt{1+\delta_{(t^q+1)k}} \sigma_{max}(\D) }{-c_0}, \\
 & c_2 = \frac{ 2\rho \mu^2(\kappa^2(1+\delta)-(1-\delta_{(t^q+1)k})) +  \sqrt{ - c_0 \rho \mu^2 (\kappa^2(1+\delta_{(t^q+1)k}) -(1-\delta_{(t^q+1)k}) ) } }{ -c_0} ,   
\end{align*}
and $\mu>0$ is a constant depending on $\rho,\kappa$.
\end{lemma}
\begin{proof}
Let $\boldsymbol{h}=\hat{\boldsymbol{\beta}} - \boldsymbol{\beta}^{*} $ and $T$ deonte the set of indices corresponding to the $k$ largest entries of $\D\bet^{*}$. Without loss of generality, assume that $(\boldsymbol{Dh})_1\geq  \ldots \geq (\boldsymbol{Dh})_n \geq 0$. Set $\alpha^q =  (||\boldsymbol{D}_{T} \boldsymbol{h}||_q^q + 2||\boldsymbol{D}_{T^c} \boldsymbol{\beta}^{*}||_q^q ) / k$. Divide $\boldsymbol{D}_{T^c} \boldsymbol{h}$ into two parts, i.e., $\boldsymbol{D}_{T^c} \boldsymbol{h} = \boldsymbol{h}_1 + \boldsymbol{h}_2$ where
\begin{align*}
 & \boldsymbol{h}_1 = \boldsymbol{D}_{T^c} \boldsymbol{h} \cdot \boldsymbol{1}_{ \{i : (\boldsymbol{D}_{T^c} \boldsymbol{h})_i > \alpha/t \} }, \\
 & \boldsymbol{h}_2 = \boldsymbol{D}_{T^c} \boldsymbol{h} \cdot \boldsymbol{1}_{ \{i : (\boldsymbol{D}_{T^c} \boldsymbol{h})_i \leq \alpha/t \} }. 
\end{align*}
Then $(\boldsymbol{D}_{T^c} \boldsymbol{h})^q = \boldsymbol{h}_1^q+\boldsymbol{h}_2^q$. And $||\boldsymbol{h}_1^q||_1 \leq ||(\boldsymbol{D}_{T^c} \boldsymbol{h})^q||_1\leq k\alpha^q$. Since all non-zero entries of $\boldsymbol{h}_1^q$ has magnitude larger than $(\alpha/t)^q$, we have
\[
 k\alpha^q \geq ||\boldsymbol{h}_1^q||_1 \geq \sum_{i\in\textrm{supp}(\boldsymbol{h}_1^q)} (\alpha/t)^q = m(\alpha/t)^q,
\]
where $m=|\textrm{supp}(\boldsymbol{h}_1^q)|$. So the sparsity of $\boldsymbol{h}_1^q$ is at most $t^qk$ (assume $t^qk$ is an integer). Thus we have
\begin{align*}
 & ||\boldsymbol{h}_2^q||_1 = ||(\boldsymbol{D}_{T^c}\boldsymbol{h})^q||_1 - ||\boldsymbol{h}_1^q||_1 \leq k\alpha^k - \frac{m\alpha^q}{t^q} = (t^qk-m)(\frac{\alpha}{t})^q, \\
 & ||\boldsymbol{h}_2^q||_{\infty} \leq (\frac{\alpha}{t})^q.
\end{align*}
According to Lemma \ref{sparsepolytope}, $\boldsymbol{h}_2^q$ lies in a polytope and can be expressed as a convex combination of sparse vectors, i.e., $\boldsymbol{h}_2^q= \sum_{i=1}^N \lambda_i \boldsymbol{u}_i$ where $\boldsymbol{u}_i$ is a $(t^qk-m)$-sparse vector. Suppose that $\mu\geq 0, c\geq 0$ are to be determined. Let $\boldsymbol{\theta}_i^q= (\boldsymbol{D}_{T} \boldsymbol{h})^q + \boldsymbol{h}_1^q + \mu\boldsymbol{u}_i$, then
\begin{equation} \label{equ1}
\begin{align*}
 \sum_{j=1}^N \lambda_j \boldsymbol{\theta}_j^q - c\boldsymbol{\theta}_i^q &= (\boldsymbol{D}_{T} \boldsymbol{h})^q + \boldsymbol{h}_1^q + \mu\boldsymbol{h}_2^q - c\boldsymbol{\theta}_i^q  \\
 &= (1-\mu-c)( (\boldsymbol{D}_{T} \boldsymbol{h})^q + \boldsymbol{h}_1^q) - c\mu\boldsymbol{u}_i + \mu (\boldsymbol{D}\boldsymbol{h})^q.
\end{align*}
\end{equation}
Since $(\boldsymbol{D}_{T} \boldsymbol{h})^q, \boldsymbol{h}_1^q, \boldsymbol{u}_i$ are $k$-, $m$-, and $(t^qk-m)$-sparse respectively, then $\boldsymbol{\theta}_i^q, \sum_{j=1}^N \lambda_j \boldsymbol{\theta}_j^q - c\boldsymbol{\theta}_i^q - \mu (\boldsymbol{D}\boldsymbol{h})^q$ are all $(t^q+1)k$-sparse vectors.

Define $\boldsymbol{\Lambda} := \textrm{diag}((\boldsymbol{Dh})_1^{1-q},(\boldsymbol{Dh})_2^{1-q},\ldots,(\boldsymbol{Dh})_p^{1-q}), \boldsymbol{B}:=\boldsymbol{X}\boldsymbol{D}^{+}\boldsymbol{\Lambda}$. Then $\boldsymbol{B}(\boldsymbol{Dh})^q=\boldsymbol{X}\boldsymbol{D}^{+}\boldsymbol{\Lambda}(\boldsymbol{Dh})^q=\boldsymbol{X}\boldsymbol{D}^{+}\boldsymbol{D}\boldsymbol{h}=\boldsymbol{Xh}$. Let $P^q = 2||\boldsymbol{D}_{T^c} \boldsymbol{\beta}^{*}||_q^q / k^{1-q/2}$, and we get
\begin{align*}
 ||\boldsymbol{\Lambda}\boldsymbol{u}_i||_2^2 &\leq \sum_{j=k+1+m}^{(t^q+1)k}  (|(\boldsymbol{Dh})_j|^{1-q} ||\boldsymbol{u}_i||_{\infty} )^2 \\
 &\leq (t^qk-m)((\frac{\alpha}{t})^{1-q} (\frac{\alpha}{t})^{q} )^2  \\
 &\leq t^{q-2}k \alpha^2 \\
 &= t^{q-2}k \left( \frac{(||\boldsymbol{D}_{T} \boldsymbol{h}||_q^q + k^{1-q/2} P^q )^{1/q}}{k^{1/q}} \right)^2  \\
 &\leq t^{q-2}k \left(  \frac{ 2^{1/q-1} (||\boldsymbol{D}_{T} \boldsymbol{h}||_q + k^{1/q-1/2} P ) }{k^{1/q}} \right)^2  \\
 &\leq t^{q-2}k \left( \frac{ 2^{1/q-1} (  k^{1/q-1/2}||\boldsymbol{D}_{T} \boldsymbol{h}||_2 + k^{1/q-1/2} P ) }{k^{1/q}} \right)^2  \\
 &= 4^{1/q-1} t^{q-2} (||\boldsymbol{D}_{T} \boldsymbol{h}||_2 + P)^2 \\
 &\leq 4^{1/q-1} t^{q-2} (||\boldsymbol{D}_{T} \boldsymbol{h} + \boldsymbol{h}_1||_2 + P)^2,
\end{align*}
where we used the inequality $[a^q+b^q]^{1/q} \leq 2^{1/q-1} [a+b]$ for $a,b\geq 0$.
Additionally, we have
\begin{align*}
& \boldsymbol{B}(\sum_{j=1}^N \lambda_j\boldsymbol{\theta}_j^q-c\boldsymbol{\theta}_i^q) =  \boldsymbol{B} [(1-\mu-c)( (\boldsymbol{D}_{T} \boldsymbol{h})^q + \boldsymbol{h}_1^q) - c\mu\boldsymbol{u}_i + \mu (\boldsymbol{D}\boldsymbol{h})^q]  \\
&~~~~~~~~~~~ = \boldsymbol{X}\boldsymbol{D}^{+}\boldsymbol{\Lambda} [(1-\mu-c)( (\boldsymbol{D}_{T} \boldsymbol{h})^q + \boldsymbol{h}_1^q) - c\mu\boldsymbol{u}_i + \mu (\boldsymbol{D}\boldsymbol{h})^q] \\
&~~~~~~~~~~~ = \boldsymbol{X}\boldsymbol{D}^{+} [(1-\mu-c)( \boldsymbol{D}_{T} \boldsymbol{h} + \boldsymbol{h}_1) - c\mu \boldsymbol{\Lambda}\boldsymbol{u}_i + \mu \boldsymbol{D}\boldsymbol{h}] \\
&~~~~~~~~~~~ = \boldsymbol{X}\boldsymbol{D}^{+} [(1-\mu-c)( \boldsymbol{D}_{T} \boldsymbol{h} + \boldsymbol{h}_1) - c\mu \boldsymbol{\Lambda}\boldsymbol{u}_i ] + \mu \boldsymbol{Xh}, \\
\end{align*}
and 
\begin{align*}
&\boldsymbol{B}\boldsymbol{\theta}_i^q = \boldsymbol{X}\boldsymbol{D}^{+}\boldsymbol{\Lambda} [(\boldsymbol{D}_{T} \boldsymbol{h})^q + \boldsymbol{h}_1^q + \mu\boldsymbol{u}_i]  \\
&~~~~~~~~~~~ = \boldsymbol{X}\boldsymbol{D}^{+} ( \boldsymbol{D}_{T} \boldsymbol{h} + \boldsymbol{h}_1 + \mu \boldsymbol{\Lambda} \boldsymbol{u}_i).
\end{align*}

Since $||\boldsymbol{D}_{T} \boldsymbol{h} + \boldsymbol{h}_1||_0 = k+m < (t^q +1) k$, we have 
\begin{align*}
& \langle  \boldsymbol{X}\boldsymbol{D}^{+} ( \boldsymbol{D}_{T} \boldsymbol{h} + \boldsymbol{h}_1) , \boldsymbol{Xh} \rangle    \leq  ||\boldsymbol{X}\boldsymbol{D}^{+} ( \boldsymbol{D}_{T} \boldsymbol{h} + \boldsymbol{h}_1) ||_2 || \boldsymbol{Xh} ||_2  \\
&~~~~~~~~~~~~~~~ \leq  \sqrt{1+\delta_{k+m}}  ||\boldsymbol{D}^{+} (\boldsymbol{D}_{T} \boldsymbol{h} + \boldsymbol{h}_1) ||_2  || \boldsymbol{Xh} ||_2  \\
&~~~~~~~~~~~~~~~ \leq  \sqrt{1+\delta_{(t^q+1)k}} \sigma_{max}(\boldsymbol{D}^{+})  || \boldsymbol{D}_{T} \boldsymbol{h} + \boldsymbol{h}_1 ||_2  || \boldsymbol{Xh} ||_2 .
\end{align*} 

Set $c=1/2$, write $\delta$ for $\delta_{(t^q+1)k}$, let the right-hand side of the equation (\ref{l2_identity}) minus the left-hand side, and obtain
\begin{align*} 
 0 &= \sum_{i=1}^N \lambda_i ||\boldsymbol{B}(\sum_{j=1}^N \lambda_j\boldsymbol{\theta}_j^q-c\boldsymbol{\theta}_i^q)||_2^2 - \sum_{i=1}^N \lambda_i (1-c)^2 ||\boldsymbol{B}\boldsymbol{\theta}_i^q||_2^2 \\
 &= \sum_{i=1}^N \lambda_i ||\boldsymbol{X}\boldsymbol{D}^{+} [(1-\mu-c)( \boldsymbol{D}_{T} \boldsymbol{h} + \boldsymbol{h}_1) - c\mu \boldsymbol{\Lambda}\boldsymbol{u}_i ] + \mu\boldsymbol{Xh} ||_2^2   \\
 &~~~~~~~~~~~~~~~~~~~~~ - \sum_{i=1}^N \lambda_i (1-c)^2 ||\boldsymbol{X}\boldsymbol{D}^{+} ( \boldsymbol{D}_{T} \boldsymbol{h} + \boldsymbol{h}_1 + \mu \boldsymbol{\Lambda} \boldsymbol{u}_i)||_2^2  \\
 &= \sum_{i=1}^N \lambda_i ||\boldsymbol{X}\boldsymbol{D}^{+} [(1-\mu-c)( \boldsymbol{D}_{T} \boldsymbol{h} + \boldsymbol{h}_1) - c\mu \boldsymbol{\Lambda}\boldsymbol{u}_i ] ||_2^2   \\
 &~~~~~~~~~~~~~~~~~~~~~ + 2 \langle \boldsymbol{X}\boldsymbol{D}^{+} [(1-\mu-c)( \boldsymbol{D}_{T} \boldsymbol{h} + \boldsymbol{h}_1) - c\mu \boldsymbol{\Lambda}\boldsymbol{h}_2^q ], \mu \boldsymbol{Xh} \rangle + \mu^2 ||\boldsymbol{Xh} ||_2^2   \\
 &~~~~~~~~~~~~~~~~~~~~~ - \sum_{i=1}^N \lambda_i (1-c)^2 ||\boldsymbol{X}\boldsymbol{D}^{+} ( \boldsymbol{D}_{T} \boldsymbol{h} + \boldsymbol{h}_1 + \mu \boldsymbol{\Lambda} \boldsymbol{u}_i)||_2^2   \\
  &\leq (1+\delta) \sum_{i=1}^N \lambda_i ||\boldsymbol{D}^{+} [ (\frac{1}{2}-\mu)( \boldsymbol{D}_{T} \boldsymbol{h} + \boldsymbol{h}_1) - \frac{\mu}{2} \boldsymbol{\Lambda}\boldsymbol{u}_i  ] ||_2^2   \\
 &~~~~~~~~~~~~~~~~~~~~~ +  \langle (1-\mu) \boldsymbol{X}\boldsymbol{D}^{+} ( \boldsymbol{D}_{T} \boldsymbol{h} + \boldsymbol{h}_1) , \mu \boldsymbol{Xh} \rangle   \\
 &~~~~~~~~~~~~~~~~~~~~~ - (1-\delta) \sum_{i=1}^N \lambda_i \frac{1}{4}  ||\boldsymbol{D}^{+}  (\boldsymbol{D}_{T} \boldsymbol{h} + \boldsymbol{h}_1   + \mu \boldsymbol{\Lambda} \boldsymbol{u}_i)  ||_2^2   \\
 &\leq (1+\delta) \sum_{i=1}^N \lambda_i  \sigma_{max}^2(\boldsymbol{D}^{+})  ||(\frac{1}{2}-\mu)( \boldsymbol{D}_{T} \boldsymbol{h} + \boldsymbol{h}_1) - \frac{\mu}{2} \boldsymbol{\Lambda}\boldsymbol{u}_i  ||_2^2   \\
 &~~~~~~~~~~~~~~~~~~~~~ +  (\mu-\mu^2) \sqrt{1+\delta} \sigma_{max}(\boldsymbol{D}^{+})  || \boldsymbol{D}_{T} \boldsymbol{h} + \boldsymbol{h}_1 ||_2  || \boldsymbol{Xh} ||_2   \\
 &~~~~~~~~~~~~~~~~~~~~~ - (1-\delta) \sum_{i=1}^N \lambda_i \frac{1}{4} \sigma_{min}^2(\boldsymbol{D}^{+}) ||\boldsymbol{D}_{T} \boldsymbol{h} + \boldsymbol{h}_1   + \mu \boldsymbol{\Lambda} \boldsymbol{u}_i  ||_2^2   
\end{align*}
\begin{equation} \label{RIP_inequ2}
\begin{align*}
 &\leq (1+\delta) \kappa^2 [ (\frac{1}{2}-\mu)^2 || \boldsymbol{D}_{T} \boldsymbol{h} + \boldsymbol{h}_1||_2^2 + \frac{\mu^2}{4} \sum_{i=1}^N \lambda_i ||\boldsymbol{\Lambda}\boldsymbol{u}_i  ||_2^2  ]  \\
 &~~~~~~~~~~~~~~~~~~~~~ +  \frac{ (\mu-\mu^2) \sqrt{1+\delta} \kappa }{\sigma_{min}(\D^{+})} || \boldsymbol{D}_{T} \boldsymbol{h} + \boldsymbol{h}_1 ||_2  || \boldsymbol{Xh} ||_2   \\
 &~~~~~~~~~~~~~~~~~~~~~ - \frac{1-\delta}{4}  (  ||\boldsymbol{D}_{T} \boldsymbol{h} + \boldsymbol{h}_1 ||_2^2  + \mu^2 \sum_{i=1}^N \lambda_i || \boldsymbol{\Lambda} \boldsymbol{u}_i  ||_2^2 )  \\
 &=  (\frac{1}{2}-\mu)^2 (1+\delta) \kappa^2 || \boldsymbol{D}_{T} \boldsymbol{h} + \boldsymbol{h}_1||_2^2 - \frac{1}{4}(1-\delta) || \boldsymbol{D}_{T} \boldsymbol{h} + \boldsymbol{h}_1||_2^2     \\ 
 &~~~~~~~~~~~~~~~~~~~~~  + (\frac{1}{4} \mu^2\kappa^2(1+\delta) - \frac{1}{4} \mu^2(1-\delta) ) \sum_{i=1}^N \lambda_i  || \boldsymbol{\Lambda}\boldsymbol{u}_i ||_2^2  \\
 &~~~~~~~~~~~~~~~~~~~~~  +  \frac{ (\mu-\mu^2) \sqrt{1+\delta} \kappa }{\sigma_{min}(\D^{+})}  ||\boldsymbol{Xh}||_2  ||\boldsymbol{D}_{T} \boldsymbol{h} + \boldsymbol{h}_1 ||_2  \\
 &\leq  (\frac{1}{2}-\mu)^2 (1+\delta) \kappa^2 || \boldsymbol{D}_{T} \boldsymbol{h} + \boldsymbol{h}_1||_2^2 - \frac{1}{4}(1-\delta) || \boldsymbol{D}_{T} \boldsymbol{h} + \boldsymbol{h}_1||_2^2     \\ 
 &~~~~~~~~~~~~~~~~~~~~~  + (\frac{1}{4} \mu^2\kappa^2(1+\delta) - \frac{1}{4} \mu^2(1-\delta) ) 4^{1/q-1} t^{q-2} (||\boldsymbol{D}_T\boldsymbol{h}+\boldsymbol{h}_1||_2 + P)^2  \\
 &~~~~~~~~~~~~~~~~~~~~~  +  \frac{ (\mu-\mu^2) \sqrt{1+\delta} \kappa }{\sigma_{min}(\D^{+})}  ||\boldsymbol{Xh}||_2  ||\boldsymbol{D}_{T} \boldsymbol{h} + \boldsymbol{h}_1 ||_2  \\
&=  \left( (\frac{1}{2}-\mu)^2 (1+\delta) \kappa^2 - \frac{1}{4}(1-\delta)  + 4^{1/q-2}t^{q-2}\mu^2(\kappa^2(1+\delta)-(1-\delta)) \right) || \boldsymbol{D}_{T} \boldsymbol{h} + \boldsymbol{h}_1||_2^2     \\ 
&~~~~~~~ + \left( 4^{1/q-3/2}t^{q-2}\mu^2(\kappa^2(1+\delta)-(1-\delta))  P + \frac{ (\mu-\mu^2) \sqrt{1+\delta} \kappa }{\sigma_{min}(\D^{+})}  ||\boldsymbol{Xh}||_2 \right) || \boldsymbol{D}_{T} \boldsymbol{h} + \boldsymbol{h}_1||_2 \\
&~~~~~~~ +  4^{1/q-2}t^{q-2}\mu^2(\kappa^2(1+\delta)-(1-\delta)) P^2.  
\end{align*}
\end{equation}
If we want to upper bound $||\boldsymbol{D}_{T} \boldsymbol{h} + \boldsymbol{h}_1 ||_2$, it requires that the coefficient of the quadratic term is less than zero, i.e.,
\begin{equation}\label{coeff2}
 (\frac{1}{2}-\mu)^2 (1+\delta) \kappa^2 - \frac{1}{4}(1-\delta)  + 4^{1/q-2}t^{q-2}\mu^2(\kappa^2(1+\delta)-(1-\delta))  < 0.
\end{equation}
If 
\begin{equation}\label{cond1_noisy}
\delta_{(t^q+1)k} < \frac{\rho(1-\kappa^4) + \kappa^2 \sqrt{4\rho+1}} { \rho(\kappa^2+1)^2 + \kappa^2 }
\end{equation}
with $\rho= 4^{1/q-2} t^{q-2}$, then there exists some $\mu$ such that the inequality (\ref{coeff2}) holds. Additionally, we have
\begin{equation}\label{cond2_noisy}
 \delta_{(t^q+1)k} < \frac{ (\rho -\rho\kappa^2 - \kappa^2)\mu^2 + \kappa^2\mu + \frac{1}{4}(1-\kappa^2) }{ (\rho + \rho\kappa^2 + \kappa^2)\mu^2 -\kappa^2\mu + \frac{1}{4} (1+\kappa^2 ) }  .
\end{equation}
To determin $\mu$, the right-hand sides of (\ref{cond1_noisy}) and (\ref{cond2_noisy}) should equal. Thus we obtain
\begin{equation}\label{mu}
 \mu = \frac{(b+1)\kappa^2 + \sqrt{ - \rho(b+1)^2\kappa^4 + (2\rho+1)(1-b^2)\kappa^2  -\rho(1-b)^2 }} { 2(\rho+1)(b+1)\kappa^2 + 2\rho (b-1) },
\end{equation}
where $\rho= 4^{1/q-2} t^{q-2}, b = \frac{\rho(1-\kappa^4) + \kappa^2 \sqrt{4\rho+1}} { \rho(\kappa^2+1)^2 + \kappa^2 }$. When $t^qk$ is not an integer, set ${t^{'}}^q = \lceil t^q k\rceil / k$. It can be verified that $b(t)$ is monotonically increasing with respect to $t$. Thus we have
\[
 \delta_{({t^{'}}^q+1)k} = \delta_{(t^q+1)k} < b(t) < b(t^{'}).
\]

Now we solve the inequality (\ref{RIP_inequ2}), and get
\begin{align*}
 & ||\boldsymbol{D}_{T} \boldsymbol{h} + \boldsymbol{h}_1 ||_2 \leq \frac{ 2\rho \mu^2(\kappa^2(1+\delta)-(1-\delta))  P + (\mu-\mu^2) \sqrt{1+\delta} \kappa / \sigma_{min}(\D^{+})  ||\boldsymbol{Xh}||_2 }{-2c_0}  \\
 & +  \frac{ \sqrt{ \left( 2\rho \mu^2(\kappa^2(1+\delta)-(1-\delta))  P + (\mu-\mu^2) \sqrt{1+\delta} \kappa/ \sigma_{min}(\D^{+})||\boldsymbol{Xh}||_2 \right)^2 - 4 c_0 \rho \mu^2 (\kappa( 1+\delta)- (1-\delta))  P^2 } }{-2c_0}   \\
 &\leq \frac{ 2\rho \mu^2(\kappa^2(1+\delta)-(1-\delta))  P + (\mu-\mu^2) \sqrt{1+\delta} \kappa/ \sigma_{min}(\D^{+}) ||\boldsymbol{Xh}||_2 }{-2c_0}  \\
 & ~~ +  \frac{  2\rho \mu^2(\kappa^2(1+\delta)-(1-\delta))  P + (\mu-\mu^2) \sqrt{1+\delta} \kappa/ \sigma_{min}(\D^{+}) ||\boldsymbol{Xh}||_2 + \sqrt{ - 4 c_0 \rho \mu^2 (\kappa(1+\delta)-(1-\delta)) } P  }{-2c_0}   \\
 & \leq c_1 \epsilon + c_2 P,
\end{align*}
where
\begin{align*}
 & c_0 = (\frac{1}{2}-\mu)^2 (1+\delta) \kappa^2 - \frac{1}{4}(1-\delta)  + \rho \mu^2(\kappa^2(1+\delta)-(1-\delta))  ,  \\
 & c_1 = \frac{ 2\kappa(\mu-\mu^2)\sqrt{1+\delta} \sigma_{max}(\D) }{-c_0}, \\
 & c_2 = \frac{ 2\rho \mu^2(\kappa^2(1+\delta)-(1-\delta))  + \sqrt{ - c_0 \rho \mu^2 (\kappa^2(1+\delta)-(1-\delta)) } }{ -c_0} .
\end{align*}

\end{proof}

\begin{theorem} \label{stable_rip}
Let $\bet^{*}\in\mathbb{R}^d$, $\y=\X\bet^{*}+\w$, and $||\w||\leq \epsilon$. Assume that $\D\in\mathbb{R}^{n\times d}$ has full column rank, and its condition number is upper bounded by $\kappa< \sqrt{\frac{2\rho+1+\sqrt{4\rho+1}}{2\rho}}$. If $\X\in\mathbb{R}^{m\times d}$ satisfies the $\mathcal{A}$-RIP over the set $\mathcal{A}=\{{\D\v}: ||\v||_0 \leq (t^q+1)k \}$ with $k, t^q k\in\mathbb{N}, t>0, q\in(0,1] $, i.e.,
\begin{equation} \label{noisy_cond}
\delta_{(t^q+1)k} < \frac{\rho(1-\kappa^4) + \kappa^2 \sqrt{4\rho+1}} { \rho (\kappa^2+1)^2 + \kappa^2 }
\end{equation}
with $\rho= 4^{1/q-2}t^{q-2}$, then the minimizer $\hat{\bet}$ of the $\ell_q$-analysis minimization problem (\ref{analysis_lq}) obeys
\begin{align*}
  & || \boldsymbol{D}\hat{\bet} - \boldsymbol{D}\bet^{*} ||_q^q \leq  2 c_1^q k^{1-q/2}  \epsilon^q +  2(2c_2^q +1) \sigma_{k}(\D\bet^{*})_q^q,  \\
  & ||\hat{\bet} - \bet^{*} ||_2 \leq  \frac{2 c_1 }{\sigma_{min}(\boldsymbol{D}) } \epsilon + \frac{ 2^{1/q}  ( 2c_2 + 1) }{\sigma_{min}(\boldsymbol{D}) } \frac{ \sigma_{k}(\D\bet^{*})_q }{k^{1/q-1/2}}  ,
\end{align*}
where
\begin{align*}
 & c_0 = (\frac{1}{2}-\mu)^2 (1+\delta) \kappa^2 - \frac{1}{4}(1-\delta_{(t^q+1)k})  + \rho \mu^2(\kappa^2(1+\delta_{(t^q+1)k})-(1-\delta_{(t^q+1)k})),  \\
 & c_1 = \frac{ 2\kappa(\mu-\mu^2)\sqrt{1+\delta_{(t^q+1)k}} \sigma_{max}(\D) }{-c_0}, \\
 & c_2 = \frac{ 2\rho \mu^2(\kappa^2(1+\delta)-(1-\delta_{(t^q+1)k})) + \sqrt{ - c_0 \rho \mu^2 (\kappa^2(1+\delta_{(t^q+1)k}) -(1-\delta_{(t^q+1)k}) ) } }{ -c_0} ,   
\end{align*}
and $\mu>0$ is a constant depending on $\rho,\kappa$.
\end{theorem}
\begin{proof}
According to Lemma \ref{tail_constraint}, we have
 \begin{align*}
  ||\boldsymbol{D}_T\boldsymbol{h}||_q \leq k^{1/q-1/2} ||\boldsymbol{D}_T\boldsymbol{h}||_2 \leq c_1 k^{1/q-1/2} \epsilon + c_2 k^{1/q-1/2} P.
 \end{align*}
Thus we get
 \begin{align*}
  ||\boldsymbol{D}_T\boldsymbol{h}||_q^q  \leq c_1^q k^{1-q/2} \epsilon^q + c_2^q k^{1-q/2} P^q =  c_1^q k^{1-q/2} \epsilon^q + 2 c_2^q ||\boldsymbol{D}_{T^c} \boldsymbol{\beta}^{*}||_q^q.
 \end{align*}
And according to Lemma \ref{cone_constraint}, we have 
\begin{align*}
& ||\boldsymbol{D}_{T^c}\boldsymbol{h}||_q^q \leq   ||\boldsymbol{D}_{T} \boldsymbol{h}||_q^q + 2||\boldsymbol{D}_{T^c} \boldsymbol{\beta}^{*}||_q^q  \\
&~~~~~~~~~~~~~ \leq c_1^q k^{1-q/2} \epsilon^q + 2(c_2^q +1) ||\boldsymbol{D}_{T^c} \boldsymbol{\beta}^{*}||_q^q.
\end{align*}
Thus we obtain
\begin{align*}
& ||\boldsymbol{Dh}||_q^q =  ||\boldsymbol{D}_T\boldsymbol{h}||_q^q + ||\boldsymbol{D}_{T^c}\boldsymbol{h}||_q^q  \\
&~~~~~~~~~~~~~ \leq 2 c_1^q k^{1-q/2} \epsilon^q + 2(2c_2^q +1) ||\boldsymbol{D}_{T^c} \boldsymbol{\beta}^{*}||_q^q .
\end{align*}

According to Lemma \ref{tail_constraint}, we have
\begin{align*}
& ||\boldsymbol{Dh}||_2 \leq ||\boldsymbol{D}_T\boldsymbol{h}||_2 + ||\boldsymbol{D}_{T^c}\boldsymbol{h}||_2 \\
&~~~~~~~~~ \leq  2 ||\boldsymbol{D}_T\boldsymbol{h}||_2  + \frac{ 2^{1/q}  ||\boldsymbol{D}_{T^c} \boldsymbol{\beta}^{*}||_q   }{k^{1/q-1/2} }  \\
&~~~~~~~~~ \leq 2 c_1 \epsilon + 2c_2 \frac{2^{1/q} ||\boldsymbol{D}_{T^c} \boldsymbol{\beta}^{*}||_q }{k^{1/q-1/2}}   + \frac{ 2^{1/q}  ||\boldsymbol{D}_{T^c} \boldsymbol{\beta}^{*}||_q   }{k^{1/q-1/2} }   \\
&~~~~~~~~~ \leq 2 c_1 \epsilon +  2^{1/q}  ( 2c_2 + 1) \frac{ ||\boldsymbol{D}_{T^c} \boldsymbol{\beta}^{*}||_q }{k^{1/q-1/2}}   . 
\end{align*}
Thus we get
\[
 ||\boldsymbol{h}||_2 \leq  \frac{2 c_1 }{\sigma_{min}(\boldsymbol{D}) } \epsilon + \frac{ 2^{1/q}  ( 2c_2 + 1) }{\sigma_{min}(\boldsymbol{D}) } \frac{ ||\boldsymbol{D}_{T^c} \boldsymbol{\beta}^{*}||_q }{k^{1/q-1/2}}    .
\]

\end{proof}

\section{The Proof of Theorem~4}
\begin{lemma} \label{commbi}
 (lemma 2.3, [4]) Let $N,s\in\mathbb{N}$ with $s<N$. There exists a family $\mathcal{U}$ of subsets of $[N]$ such that
 \begin{enumerate}
  \item [(i)] Every set in $\mathcal{U}$ consists of exactly $s$ elements.
  \item [(ii)] For all $I,J\in\mathcal{U}$ with $I\neq J$, it holds $|I\cap J|<s/2$.
  \item [(iii)] The family $\mathcal{U}$ is ``large'' in the sense that 
  \[
   |\mathcal{U}| \geq \left(\frac{N}{4s}\right)^{s/2} .
  \]
 \end{enumerate}
\end{lemma}

\begin{theorem}
 Let $m,n,k\in\mathbb{N}$ with $m,k<n$. Suppose that a matrix $\X\in\mathbb{R}^{m\times d}$, a linear operator $\D\in\mathbb{R}^{n\times d}$ and a decoder $\bigtriangleup :\mathbb{R}^m\rightarrow \mathbb{R}^d$ solving $\y=\X\bet$ satisfy for all $\bet\in\mathbb{R}^d$,
 \[
  ||\D\bet- \bigtriangleup (\X\bet)||_q^q \leq C \sigma_{k}(\D\bet)_q^q
 \]
with some constant $C>0$ and some $0<q\leq 1$. Then the minimal number of samples $m$ obeys
\[
 m \geq C_1 q k \log (n/4k)
\]
with $k=||\D\bet||_0$ and $C_1=1/(2\log(2C+3))$.
\end{theorem}
\begin{proof}
Consider a family of $\mathcal{U}$ of subsets of $[N]$ satisfying $(i),(ii),(iii)$ of Lemma \ref{commbi}. For each $I\in\mathcal{U}$, we define a $s$-sparse vector $(\boldsymbol{D\beta})_I$ with $||(\boldsymbol{D\beta})_I||_q=1$ by 
\[
 (\boldsymbol{D\beta})_I := \frac{1}{s^{1/q}} \sum_{i\in I} \boldsymbol{e}_i .
\]
According to (ii), we have
\[
 ||(\boldsymbol{D\beta})_I - (\boldsymbol{D\beta})_J||_q^q > \frac{2s-2s/2}{s} = 1
\]
for $I,J\in\mathcal{U}, I\neq J$. With $\rho:=(2(C+1))^{-1/q}$, we claim that $\{ \boldsymbol{X}\boldsymbol{D}^{+} ( (\boldsymbol{D\beta})_I + \rho B_q^n), I\in\mathcal{U} \}$ is a disjoint collection of subsets of $\boldsymbol{X}\boldsymbol{D}^{+}(\mathbb{R}^N)$ with algebraic dimension $r\leq m$. Suppose indeed that there exist $I,J\in\mathcal{U}$ with $I\neq J$ and $\boldsymbol{z},\boldsymbol{z}^{'}\in\rho B_q^n$ such that $\boldsymbol{X}\boldsymbol{D}^{+}((\boldsymbol{D\beta})_I+\boldsymbol{z})=\boldsymbol{X}\boldsymbol{D}^{+}((\boldsymbol{D\beta})_J+\boldsymbol{z}^{'})$. Define $\bigtriangleup (\boldsymbol{X\beta}) = \boldsymbol{D} \bigtriangleup_0 (\boldsymbol{y})$ with $\bigtriangleup_0 (\boldsymbol{y}) = \argmin_{\boldsymbol{\beta}: \boldsymbol{y}=\boldsymbol{X\beta}} ||\boldsymbol{D\beta}||_q^q$. A contradiction follows from
\begin{align*}
 & 1 < ||(\boldsymbol{D\beta})_I - (\boldsymbol{D\beta})_J||_q^q  \leq ||(\boldsymbol{D\beta})_I+\boldsymbol{z} - \bigtriangleup (\boldsymbol{X}\boldsymbol{D}^{+}((\boldsymbol{D\beta})_I +\boldsymbol{z}) ) ||_q^q + ||\boldsymbol{z}||_q^q \\
 &~~~~~~~~~~~~~~~~~~  + ||(\boldsymbol{D\beta})_J+\boldsymbol{z}^{'} - \bigtriangleup (\boldsymbol{X}\boldsymbol{D}^{+} ((\boldsymbol{D\beta})_J +\boldsymbol{z}^{'})) ||_q^q  + ||\boldsymbol{z}^{'}||_q^q \\
 &~~~~~~~~~~ \leq C\sigma_s((\boldsymbol{D\beta})_I+\boldsymbol{z}) + C\sigma_s((\boldsymbol{D\beta})_J+\boldsymbol{z}^{'}) + ||\boldsymbol{z}||_q^q + ||\boldsymbol{z}^{'}||_q^q \\
 &~~~~~~~~~~ \leq C||\boldsymbol{z}||_q^q + C||\boldsymbol{z}^{'}||_q^q  + ||\boldsymbol{z}||_q^q + ||\boldsymbol{z}^{'}||_q^q \leq 2(C+1)\rho^q = 1.
\end{align*}
Since the pseudo norm $||\cdot||_q^q$ satisfies the triangle inequality, we can observe that the collection $\{ \boldsymbol{X}\boldsymbol{D}^{+} ( (\boldsymbol{D\beta})_I + \rho B_q^n), I\in\mathcal{U} \}$ is contained in $(1+\rho^q)^{1/q} \boldsymbol{X}\boldsymbol{D}^{+} (B_q^n)$. Thus we have
\begin{align*}
& |\mathcal{U}| \rho^r \textrm{vol} ( \boldsymbol{X}\boldsymbol{D}^{+} (B_q^n) ) = \sum_{I\in\mathcal{U}} \textrm{vol} ( \boldsymbol{X}\boldsymbol{D}^{+} ( (\boldsymbol{D\beta})_I + \rho B_q^n) )  \\
&~~~~~~~~~~~~~~~~~ \leq \textrm{vol}  ( (1+\rho^q)^{1/q} \boldsymbol{X}\boldsymbol{D}^{+} (B_q^n) ) = (1+\rho^q)^{r/q} \textrm{vol} ( \boldsymbol{X}\boldsymbol{D}^{+} (B_q^n) ) .
\end{align*} 
According to $(iii)$, we get
\[
 \left(\frac{n}{4k}\right)^{k/2}  \leq (\rho^{-q}+1)^{r/q} \leq (\rho^{-q}+1)^{m/q} =(2C+3)^{m/q}.
\]
So we have
\[
 m \geq C_1 q k \log (n/4k)  ~~~~~ \textrm{with}~~ C_1=1/(2\log(2C+3)).
\]

\end{proof}

\section{The Proof of Theorem~5}

\begin{proof}
 Let $L:=(t+1)S_1$, $ \widetilde{l}^q=  t^{\frac{q}{2-q}}$, $\widetilde{S}_q=\frac{L}{ \widetilde{l}^q+1}$, then
\begin{align*}
\delta_{( \widetilde{l}^q+1)\widetilde{S}_q} = \delta_{(t+1)S_1} &\leq \frac{ \frac{1}{4} t^{-1} (1-\kappa^4) - \kappa^2 \sqrt{ 4*\frac{1}{4} t^{-1} + 1} } { \frac{1}{4} t^{-1} (\kappa^2 +1)^2 + \kappa^2}   \\
& = \frac{ \frac{1}{4} \widetilde{l}^{q-2} (1-\kappa^4) - \kappa^2 \sqrt{ 4* \frac{1}{4} \widetilde{l}^{q-2} +1 }} {\frac{1}{4} \widetilde{l}^{q-2} (\kappa^2 +1)^2 + \kappa^2}.
 \end{align*}
Since the map $\widetilde{l} \mapsto \frac{ \frac{1}{4} \widetilde{l}^{q-2} (1-\kappa^4) - \kappa^2 \sqrt{ 4* \frac{1}{4} \widetilde{l}^{q-2} +1 }} {\frac{1}{4} \widetilde{l}^{q-2} (\kappa^2 +1)^2 + \kappa^2}$
 is increasing in $\widetilde{l}$ for $\widetilde{l}> 0$, we only need choose $l $ and $S_q$ such that 
 $l^q \geq \widetilde{l}^q$ and $(l^q +1)S_q=L$ ($l^q, S_q\in\mathbb{N}$).
Recall that $S_1\in\mathbb{N}$ and $(t+1)S_1\in\mathbb{N} ~(t\geq 1)$, thus we have
 $1\leq \widetilde{l}^q < t \leq L-1,$ and $\widetilde{l}^q \in\{ \frac{n}{L-n}:n=\lceil\frac{L}{2}\rceil,...,L-1 \}$.
 Thus there exists $n^*$ such that 
 \[
  \frac{n^*-1}{L-n^*+1}< \widetilde{l}^q \leq \frac{n^*}{L-n^*},
 \]
 and $\frac{n^*}{L-n^*}>1$. As a result, we have  
 \[
  L-n^*\leq \widetilde{S}_q<L-n^*+1.
 \]
 Therefore we can choose $l^q =\frac{n^*}{L-n^*}$ and 
 \[
 S_q=\lfloor \widetilde{S_q}\rfloor= \Big\lfloor \frac{t+1}{ t^{\frac{q}{2-q}} +1} S_1  \Big\rfloor 
 \]
 such that
 \[
  \delta_{(l^q +1)S_q} <  \frac{ \frac{1}{4} \widetilde{l}^{q-2} (1-\kappa^4) - \kappa^2 \sqrt{ 4* \frac{1}{4} \widetilde{l}^{q-2} +1 }} {\frac{1}{4} \widetilde{l}^{q-2} (\kappa^2 +1)^2 + \kappa^2} . 
 \]
\end{proof}

\section{The Proof of Lemma~1}
\begin{proof}
First we have
\begin{align*}
 ||\boldsymbol{D}\bet^{(k)}||_q^q &\leq \sum_{i=1}^n [ |\boldsymbol{D}_{i.}\bet^{(k)}|^{\alpha} + (\varepsilon^{(k)})^{\alpha} ]^{ \frac{q}{\alpha} } 
\leq F(\bet^{(k)},\varepsilon^{(k)})/\lambda \leq 
F(\bet^{(0)},\varepsilon^{(0)})/\lambda.
\end{align*}
Thus we get
\[ 
||\boldsymbol{D}\bet^{(k)}||_2 \leq ||\boldsymbol{D}\bet^{(k)}||_q \leq 
 ( F(\bet^{(0)},\varepsilon^{(0)})/\lambda )^{1/q} .
\]
Since it is assumed that $\boldsymbol{D}$ has full column rank, we have
\begin{align*}
 ||\bet^{(k)}||_2 = ||\boldsymbol{D}^{+}\boldsymbol{D}\bet^{(k)}||_2 
 \leq \sigma_{max}(\boldsymbol{D}^{+}) ||\boldsymbol{D}\bet^{(k)}||_2 
 \leq \sigma_{min}^{-1}(\boldsymbol{D}) ( F(\bet^{(0)},\varepsilon^{(0)})/ \lambda )^{1/q}.
\end{align*}
For $F(\bet^{(k+1)},\varepsilon^{(k+1)}) \leq F(\bet^{(k)},\varepsilon^{(k)})$, we consider
\begin{align*}
&  F(\bet^{(k+1)},\varepsilon^{(k+1)}) = \frac{1}{2}||\boldsymbol{y}-\X \bet^{(k+1)}||_2^2 
+ \min_{\boldsymbol{\eta}>\boldsymbol{0}} \Big\{ \frac{\lambda q}{\alpha} \sum_{i=1}^n  \Big[ \eta_i\Big({|\boldsymbol{D}_{i.} \bet^{(k+1)}|}^{\alpha}+{\varepsilon^{(k+1)}}^{\alpha}\Big) + \frac{\alpha-q}{q} \frac{1}{{\eta_i}^{\frac{q}{\alpha-q}}} \Big] \Big\} \\
&~~~ \leq \frac{1}{2}||\boldsymbol{y}-\X \bet^{(k+1)}||_2^2 
 +  \frac{\lambda q}{\alpha} \sum_{i=1}^n  \Big[ \eta_i^{(k+1)}\Big( {|\boldsymbol{D}_{i.} \bet^{(k+1)}|}^{\alpha}+{\varepsilon^{(k+1)}}^{\alpha}\Big) + \frac{\alpha-q}{q} \frac{1}{{{\eta_i}^{(k+1)}}^{\frac{q}{\alpha-q}}} \Big] \\
&~~~ = \min_{\bet} \Big\{ \frac{1}{2}||\boldsymbol{y}-\X \bet||_2^2 
 +  \frac{\lambda q}{\alpha} \sum_{i=1}^n  \Big[ \eta_i^{(k+1)}\Big( {|\boldsymbol{D}_{i.} \bet|}^{\alpha}+{\varepsilon^{(k+1)}}^{\alpha}\Big) + \frac{\alpha-q}{q} \frac{1}{{{\eta_i}^{(k+1)}}^{\frac{q}{\alpha-q}}} \Big] \Big\}   \\
&~~~ \leq \frac{1}{2}||\boldsymbol{y}-\X \bet^{(k)}||_2^2 
 +  \frac{\lambda q}{\alpha} \sum_{i=1}^n  \Big[ \eta_i^{(k+1)}\Big( {|\boldsymbol{D}_{i.} \bet^{(k)}|}^{\alpha}+{\varepsilon^{(k+1)}}^{\alpha}\Big) + \frac{\alpha-q}{q} \frac{1}{{{\eta_i}^{(k+1)}}^{\frac{q}{\alpha-q}}} \Big]   \\
&~~~ = \frac{1}{2}||\boldsymbol{y}-\X \bet^{(k)}||_2^2 
 + \min_{\boldsymbol{\eta}>\boldsymbol{0}} \Big\{ \frac{\lambda q}{\alpha} \sum_{i=1}^n  \Big( \eta_i({|\boldsymbol{D}_{i.} \bet^{(k)}|}^{\alpha}+{\varepsilon^{(k)}}^{\alpha}) + \frac{\alpha-q}{q} \frac{1}{{\eta_i}^{\frac{q}{\alpha-q}}}  \Big) \Big\}  \\
&~~~ = F(\bet^{(k)},\varepsilon^{(k)}).
\end{align*}

\end{proof}

\section{The Proof of Theorem~6}
\begin{proof}
Since $\X$ satisfies the $\mathcal{A}$-RIP over the set $\mathcal{A}=\{{\bet}:\D_{\Lambda}{\bet}=\boldsymbol{0},|\Lambda|\geq l\}$ of order $2l-p$, we have
\begin{align*}
 ||\hat{\bet}-\bet^{*}||_2 &\leq  \frac{1}{\sqrt{1-\delta_{2l-p}}} ||\X(\hat{\bet}-\bet^{*})||_2 \\
 &= \frac{1}{\sqrt{1-\delta_{2l-p}}} ||\boldsymbol{y}-\X\hat{\bet} - (\boldsymbol{y}-\X\bet^{*})||_2 \\
 &\leq \frac{1}{\sqrt{1-\delta_{2l-p}}} (||\boldsymbol{y}-\X\hat{\bet}||_2 + ||\boldsymbol{y}-\X\bet^{*}||_2 )  \\
 &= \frac{1}{\sqrt{1-\delta_{2l-p}}} (||\boldsymbol{y}-\X\hat{\bet}||_2 + ||\w||_2). 
\end{align*}
In addition, we have
\begin{align*}
||\boldsymbol{y}-\X\hat{\bet}||_2 &\leq\sqrt{2 F(\hat{\bet},\hat{\varepsilon})} \leq\sqrt{2 F(\bet^{(0)},\varepsilon^{(0)})} 
= \sqrt{2\lambda\sum_{i=1}^n (|\boldsymbol{D}_{i.}\bet^{(0)}|^{\alpha} + (\varepsilon^{(0)})^{\alpha})^{q/\alpha}}. 
\end{align*}
Combining the above two inequalities, we get
\begin{align*}
 ||\hat{\bet}-\bet^{*}||_2 \leq 
\frac{1}{\sqrt{1-\delta_{2l-p}}} \sqrt{2\lambda\sum_{i=1}^n (|\boldsymbol{D}_{i.}\bet^{(0)}|^{\alpha} + {\varepsilon^{(0)}}^{\alpha} )^{q/\alpha} } 
 + \frac{1}{\sqrt{1-\delta_{2l-p}}} \epsilon.
\end{align*}

\end{proof}